\documentclass[12pt]{article}

\usepackage{color}

\usepackage{graphicx}
\usepackage{amssymb}
\usepackage{fancyhdr}

\usepackage{datetime}

\renewcommand{\theequation}{\arabic{section}.\arabic{equation}}
\newcommand{\cf}{\emph{cf}}
\newcommand{\e}{\mathrm{e}}

\newcommand{\R}{\mathbb{R}}

\newcommand{\Hm}[1]{\leavevmode{\marginpar{\tiny%
$\hbox to 0mm{\hspace*{-0.5mm}$\leftarrow$\hss}%
\vcenter{\vrule depth 0.1mm height 0.1mm width \the\marginparwidth}%
\hbox to
0mm{\hss$\rightarrow$\hspace*{-0.5mm}}$\\\relax\raggedright #1}}}
\newtheorem{claim}{Claim}[section]
\newtheorem{theorem}[claim]{Theorem}
\newtheorem{lemma}[claim]{Lemma}

\newtheorem{proposition}[claim]{Proposition}
\newtheorem{corollary}[claim]{Corollary}
\newtheorem{remark}[claim]{Remark}
\newtheorem{remarks}[claim]{Remarks}
\newtheorem{definition}[claim]{Definition}

\newtheorem{example}[claim]{Example}

\newenvironment{proof}[1][Proof]{\textsl{#1.} }{\ \rule{0.4em}{0.7em}}
\begin{document}

%\title{On some spectral properties of a quantum system
%with double line interaction}

\title{\bf Renormalized von Neumann entropy with application to entanglement
in genuine infinite dimensional systems}
\author{Roman Gielerak%$^a$, %Sylwia Kondej$^b$
}
\date{
\small \emph{
\begin{center}
%\begin{itemize}
%\item[$a)$]
Institute of  Control \& Computation  Engineering, \\
University of Zielona G\'ora, ul.\ Szafrana 2, \\ 65-246 Zielona G\'ora, Poland; \\
e-mail: r.gielerak@issi.uz.zgora.pl
%\item[$b)$] Institute of Physics,
%University of Zielona G\'ora, ul.\ Szafrana 2, 65-246 Zielona G\'ora, Poland;
%\end{itemize}
\end{center}
}
\medskip
\today} %,  \currenttime}
\maketitle
\begin{abstract}
A renormalized  version of  the  von Neumann quantum  entropy (which  is  finite and  continuous   in general, infinite  dimensional  case) and   which  obeys several of  the natural physical demands (as  expected  for  a  “good”  measure of  entanglement  in  the  case  of  general quantum  states  describing  bipartite and  infinite-dimensional systems) is  proposed. The  renormalized  quantum entropy  is  defined  by the explicit  use  of  the  Fredholm  determinants    theory. To  prove   the main results on continuity  and  finiteness of  the introduced  renormalization the  fundamental  Grothendick  approach, which  is  based  on  the  infinite dimensional  Grassmann algebra  theory, is  applied. Several  features  of  majorization  theory  are  preserved  under  the  introduced  renormalization as  it is  proved in  this  paper. This  fact enables  us  to extend most of  the   known (mainly, in  the  context of  two-partite, finite-dimensional  quantum systems) results of  the  LOCC  comparison theory to the  case  of  genuine  infinite-dimensional , two-partite quantum  systems.
\end{abstract}
%

%\HD{I changed title, added authors and affiliations,
%changed abstract. OK?}

\newpage

%\end{document}
%---------------------%
\section{Introduction}
%---------------------%
%
%\HD{I partly rewrote the introduction,
%added some references. Please check.}%
%

Let us  consider  the model of two spinless  quantum  particles interacting  with each other and living  in  three dimensional  Euclidean  space $\R^3$.
%\end{document}
%As everyone of us  learned  from  the  basic  university  course on  quantum  mechanics, see i.e. [  1,2     ] ,
Generally, the  states  of  such   quantum  systems  are described by  the  density  matrices which  are  nonnegative, of trace  class  operators acting  on  the  space  $ \mathcal H=L_2 ( \R^3  ) \otimes  L_2 ( \R^3 )$, see~\cite{F-lectures, LL}.
 The latter is, in fact, unitary equivalent to $L_2 ( \R^6 )$.
In  particular, any pure state can be represented  (up  to the  global phase  calibration) by  the  corresponding  wave function $\psi (x,y)\in \mathcal H$; then the density matrix takes the form of the projector onto the ket vector $ |\psi \rangle $.

Using Schmidt decomposition theorem, \cf~\cite[Thm. 26.8]{BB}, we conclude that for any pure normalized  state $\psi \in \mathcal  H$  there exist:
a sequence  of  nonnegative numbers $\{\lambda _n\}_{n=1}^\infty$ (called  the  Schmidt  coefficients  of $\psi$) satisfying the condition $\sum_{n=1}^\infty \lambda_n^2 =1$
and two complete orthonormal  systems  of  vectors $\{\varphi_n \}_{n=1}^\infty$, $\{\omega_n \}_{n=1}^\infty$ in $L_2 ( \R^3 )$  such that  the  following  equality (in  the  $L_2$ – space  sense)
\begin{equation}\label{Schmidt}
  \psi (x,y)= \sum_{n=1}^\infty \lambda_n \varphi_n (x)\omega_n (y)
\end{equation}
holds.

In  particular, we call the  vector $\psi $ %given by (\ref{Schmidt})
a separable pure state \emph{iff}  there appears only one non-zero  Schmidt  coefficient in the decomposition~(\ref{Schmidt}).
If  the  number of non-zero Schmidt coefficients is finite than we say that $\psi $ is   of finite Schmidt rank pure state.  In this case, one can apply the standard and the most  frequently used  measure of amount of entanglement included in the state $\psi$ which is given by the von Neumann formula
\begin{equation}\label{eq-EN}
  \mathrm{EN}(\psi) = - \sum_{n=1}^\infty \lambda_n ^2 \log (\lambda_n ^2)\,.
\end{equation}
 Although, the  set of finite Schmidt rank pure  states  of  the  system  under consideration is  dense (in  the  $L_2$-topology) on the  corresponding  Bloch  sphere (this  time  infinite-dimensional and  given here  modulo  global phase  calibration for simplification of  the  following  discussion only) denoted as $B= \{\psi \in L_2(\R^6)\, :\,  \|\psi  \|=1\}$, it  appears that also  the  set  of  infinite Schmidt rank   pure states is dense there. The  situation  is  even more  complicated  as  it can be shown  that  the  set  of  pure  states for  which  the  value of von Neumann entropy %$\mathrm{EN}$
 is  finite is  dense in $B$ but  also  the  set  of  states  with infinite entropy of  entanglement   is   dense in  this  Bloch  sphere \cite{ESP}.

 Similar  results  on densities  of  the   infinite/finite Schmidt rank   states  are  also  valid  in the   proper physical  $L_1$-topologies on  the corresponding Bloch  sphere. Very roughly, the reason  is  that  in  infinite  dimensions  there  are  many (too many in  fact)  sequences  $\{\lambda_n$\})  such that: for  all  $n$, $\lambda_n\geq 0$   and $\sum_{n=1}^\infty \lambda_n^2 =1$   but $ \sum_{n=1}^\infty  \lambda_n ^2 \log (\lambda_n ^2 )=-\infty$. In  other  words, the  set  of  pure states  for  which  the  entropy  is finite  has  no  internal  points  and  this  fact  causes  serious  problems in the fundamental question on  continuity  of  the von Neumann entropy in genuine   infinite  dimensional  setting~\cite{ESP, EisertPlenio2003}.  In finite dimensions the von Neumann entropy is a non-negative, concave, lower semi-continuous and also norm continuous  function  defined on the set of all quantum states.
 A  lot  of  fundamental  results  on  several quantum versions of  entropy, in  particular, on  von Neumann  entropy have  been  obtained  in  the  last  decades, \cf~\cite{OP, WGeneral, PQuantum, LR, SConvergence, FContinuity, UEntropy, Chehade26, 23}. However, in   the infinite  dimensional setting,  the  conventionally  defined  von  Neumann  entropy  is  taking the value $+\infty $ on a dense subset of the space  of  quantum  states  of  the  system  under  consideration~\cf~\cite{ESP, WGeneral, BB, HS, SH08, S06, S10, S15, S15a, EisertPlenio2003}.

Nevertheless,  defined  in the  standard  way von Neumann entropy has continuous  and bounded restrictions to some special (selected  by  some  physically motivated arguments) subsets of quantum states. For example,  the set of states of the system of quantum oscillators with bounded   mean energy  forms  a set of states  with  finite  entropy~\cite{ESP, EisertPlenio2003, Vedral2001, Tomamichel2016}. Since the continuity of the entropy is a very desirable property in the analysis of quantum systems, various , sufficient for continuity conditions have been obtained up to now. The earliest one, among them, seems to be Simon’s dominated convergence theorems presented in~\cite{LR, SConvergence, FContinuity} and widely used in applications, see~\cite{OP, WGeneral, PQuantum}.
Another  useful continuity condition originally appeared in~\cite{ESP, EisertPlenio2003} and can be formulated as the continuity of the entropy on each subset of states characterized by bounded mean value of a given positive unbounded operator with discrete spectrum, provided , that its sequence of eigenvalues has a sufficient large rate of decrease. Some special conditions yielding the continuity of the von Neumann entropy are formulated in the series of  papers by Shirokov~\cite{HS, SH08, S06, S10, S15, S15a}. A stronger version of the stability property of the set of quantum states naturally called  there as strong stability was introduced by Shirokov  together  with some applications  concerning the problem of approximations of concave (convex) functions on the set of quantum states and a new approach to the analysis of continuity of such functions has been presented  there. Several other attempts  and  ideas  to  deal  with  the  noncommutative, infinite  dimensional  setting were  published  in  the current literature also. Some  of  them are based, on a very  sophisticated,  tools and  methods, such as, for  example theory of noncommutative (versions  of) the (noncommutative)  log-Sobolev  spaces of  operators \cite{Madore1999}.

\subsection{The main idea of the paper}

The main  idea of the  present paper is  to introduce an appropriate  renormalized  version  of  the  widely known  von Neumann formula for  the  entropy  in the non commutative  setting \cite{OP, WGeneral, PQuantum}. The notion of von Neumann  entropy is  one of the basic concept introduced  and  applied   in quantum  physics .However  formula proposed  by  von Neumann   works  perfectly  well  only in  the  context of  finite  dimensional  quantum systems \cite{ESP, EisertPlenio2003}. The  extension to the genuine  infinite-dimensional  setting is  not  straightforward  and  meets  several serious  obstacles as mentioned  in the  previous sentences. Our prescription for extracting finite part of the (otherwise  typically in the sense of Baire category theory) infinite  valued standard  von Neumann formula is  very simple. For this   goal,  let  Q  be  a  quantum state ,i.e.   Q  is  non-negative, of trace  class operator defined  on some  separable  Hilbert space $\mathcal{H}$  and  such that  $\textrm{Tr}(Q) = 1$. The  standard  definition  of  von Neumann  entropy $\mathrm{EN}$ is given as :
\begin{equation}
	\mathrm{EN}(Q) = -\textrm{Tr}( Q \log(Q) ) 	
\end{equation}
Our renormalisation  proposal, denoted  as  FEN, is  given by:
\begin{equation}
\mathrm{FEN}(Q) =  \textrm{Tr} ( Q+1_\mathcal{H} ) \log( Q + 1_\mathcal{H})
\end{equation}
where $1_\mathcal{H}$ stands for  the unit operator in $\mathcal{H}$.

\begin{claim}
For  any  such  $Q$  the  value  $\mathrm{FEN}(Q)$ is  finite  .
\end{claim}
\begin{proof}
Let  $\sigma(Q) = ( \tau_1, \ldots, \tau_n, \ldots )$  be  sequence  representing the  spectrum  of  Q  and  ordered  in non-increasing  order  (and  with  multiplicities  included). Using  the  elementary inequality
\begin{equation}
\log(  1+x) \leq x \;\;\; \mathrm{for} \;\;\;  x  \geq  0,
\end{equation}
together with functional calculus \cite{RSI, BEH, BB} we have the  following estimate  
\begin{equation}
\begin{array}{lcl}
\mathrm{FEN}(Q) & = & \sum_{n=1}^{\infty} (\tau_n + 1) \log (1 + \tau_n) \\ \\
				&   & \leq \sum_{n=1}^{\infty} ( \tau_n^2 + \tau_n ) \\ \\
				&   & \leq 2 \cdot \sum_{n=1}^\infty \tau_n \leq 2 .
\end{array}
\end{equation}
\end{proof}

This means  that  the  introduced  map
\begin{equation}
\mathrm{FEN} : E(\mathcal{H}) \mapsto [0,\infty)       
\end{equation}
is  finite  on  the  space $E(\mathcal{H})$ of the  quantum states on $\mathcal{H}$. The detailed mathematical study of the basic properties of the introduced  here renormalizations of  von  Neumann  entropy is  the main  topic of the present  paper. Additionally presentations of   several  applications of the introduced entropy FEN  and  adressed to the Quantum Information Theory \cite{NC, BZ, Vedral2001, Tomamichel2016} are included also. 
To  achieve  all these  goals  the  theory of  Fredholm  determinants as given  by  Grothendick \cite{G56} is intensively used in  the following  below  presentation.  Also  certain  results  from  the infinite  dimensional  majorisation  theory theory \cite{27, 28, 29, 31, Kaftal2010}  have been  used. Some, a very preeliminary and  illustrative idea of von Neumann entropy  renormalization was  recently  published  by  the  Author  in \cite{Sawerwain2021}.

\subsection{Organisation  of  the  paper}

In  the  next  Section 2 the  technique of the  Fredholm  determinants is  successfully  applied to  show that the proposed  here, renormalized version of von  Neumann  entropy  formula in the  genuine infinite-dimensional  setting is  finite  and  continuos (in the $L_1$-topology meaning) on the  space of  quantum  states. Elements of the so-called  multiplicative  version of  the standard  majorization  theory \cite{NC, BZ, 27, 28, 16} is being  introduced in Section 3. The main  results  reported  there are: the  rigorous  proof of  monotonicity of  the  introduced  renormalization of  von Neumann  entropy under  the  semi-order  relations (caused by  the defined there multiplicative majorization) lifted  to the space  of  quantum  states.  Additionally an  extension  of the basic  (in the  present  context) Alberti-Uhlmann theorem \cite{27} is proved in Section~\ref{lbl:sec:majorisation:theory}. Also monotonicity of the  introduced  notion  of renormalized  von Neumann  entropy under the  action  of  a general  quantum operations on  quantum states  is proved  there.  Section~\ref{lbl:sec:tensor:states} is devoted  to the  study of  two-partite quantum systems of infinite  dimensions  both  (the  case  of  one  factor being  finite dimensional is analysed in  details see \cite{GierelarkFuture1, RGielerakFuture5}).  In  particular,  the  corresponding  reduced  density matrices  are  studied  there  and  some  useful  formula  and  estimates  of  the  corresponding  renormalized  entropies  are  included  there. The particular  case  of pure  bipartite  states is  analysed  from  the  point  of  view  of majorization  theory  with  the  use of  novel, local unitary  and  monotonous invariants  perspective  of  Gram  operators as introduced  in another papers \cite{Gielerak2020, PPAM2022, Gielerak2021, RGielerakFuture4, Gielerak2022}. The  finite  dimensional results  of this type,  presented  in \cite{Gielerak2020, PPAM2022, RGielerakFuture4} are being  extended  to the infinite dimensional setting there with the use of Fredholm  determinants  theory \cite{RGielerakFuture4}. At  the  end  of  this  paper three appendices  are  attached  to make  this  paper  autonomous and  also because  some  additional  results  which might be  helpful in  further  developments of the  ideas presented  here are being  formulated  there. In appendix~\ref{lbl:app:sec:fock:spaces} the  Author have  presented (after  Grothendik \cite{G56}, see also  B.Simon \cite{S05})  crucial  facts and  estimates  from  the infinite  dimensional Grassman algebra  theory  with  the  applications  to control  Fredholm  determinants.  Appendix  B includes  several  results and  formulas on the different  types of  combined  Schmidt and spectral  decompositions  of a general bipartite quantum states. Finally in Appendix C some useful remarks  on the  operator  valued  function  $\log(1 + Q)$ are  collected.

Extensions  of  the  approach  to  the  renormalisation of  the  von Neumann  entropy presented in this paper  to a very  rich  palette of  intriguing  questions, like  for  example  renormalisation of quantum relative entropy and quantum relative information notions \cite{10, 11, 14, 15, 17, 18,  GierelarkFuture2,  Vedral2001, Tomamichel2016} are also visible for the Author and some work on them is in progress.

\section{Renormalized version of  the von Neumann   entropy} \label{lbl:sec:Renorm:Neumann:Entropy}

\subsection{Some mathematical notation} 

Assume that   $\mathcal H$ is a separable infinite dimensional Hilbert space\footnote{The results of this paper hold for finite dimensional Hilbert spaces as well.}. In this paper we use the following standard notation:

\begin{itemize}
\item%[]
$L_1 (\mathcal{H} )$ stands for    the  Banach  space  of  trace  class operators  acting on $\mathcal{H}$  and equipped with  the  norm $\|Q \|_1 =\textrm{Tr}\, [ |Q^{\dag } Q|^{1/2 }]$, where $Q\in L_1 (\mathcal{H} )$ and the symbol $\dag $ means  the  hermitian conjugation,

\item%[(ii)]
 $L_2(\mathcal{H})$  denotes the  Hilbert-Schmidt class of  operators acting in $\mathcal{H}$ and with the  scalar  product  $\langle Q| Q' \rangle_{HS}= \textrm{Tr}\,  [ Q^\dag Q' ]$, where $Q\in L_2(\mathcal{H})$,

\item%[(ii)]
 $B(\mathcal{H})$  denotes the space of the bounded operators with norm defined as the operator norm $\|\cdot \|$,%; in the literature the space $L_\infty (\mathcal{H})$ is also denoted as $B(\mathcal H )$ \Hm{one notation please, let's remove $B(H)$ and $\|\cdot \|_\infty $}
  %Hilbert-Schmidt class of  operators acting in $\mathcal{H}$ and with the  scalar  product  $\langle Q| Q' \rangle_{HS}= \textrm{Tr}\,   Q^\dag Q' $

\item%[(ii)]
Let  $E(\mathcal H)$  be  the  complete metric   space of  quantum states $Q$ on  the  space  $\mathcal H$, i.e.  the  $L_1$-completed  intersection  of the  cone of nonnegative, trace  class  operators on  $\mathcal H$ and  $L_1$-closed  hyperplane    described by the  normalisation condition  $\textrm{Tr} [Q]=1$.
\end{itemize}

In further discussion we will relay on the following inequalities, \cf~\cite{RSI, BEH},
\begin{equation}\label{eq-ineq2.1a}
\|A B\|_1 \leq \|A\|_1 \|B\|_1\,,\quad A, B \in L_1 (\mathcal H ),
\end{equation}
\begin{equation}\label{eq-ineq2.1b}
\|A B\|_1 \leq \|A\| \|B\|_1\,,\quad A \in B (\mathcal H)\,, B \in L_1 (\mathcal H );
\end{equation}
the latter inequality also holds for $\|B A \|_1$ with obvious changes.
\\ \\
The following  spaces of sequences will be used  in further analysis
%\item%[(ii)]
%For infinite sequence  $\underline{a}= (a_1,...,a_i,...)\,, a_i \in \R$
\begin{equation}\label{eq-def1}
C^ \infty = \{ \underline{a}= (a_1,...,a_n,...)\,, a_n \in \R \}\,,
\end{equation}
\begin{equation}\label{eq-def2}
C^\infty_+ = \{ \underline{a}\in C^\infty  \,:\, a_n \geq 0 \} \,,
\end{equation}
\begin{equation}\label{eq-def3}
C^\infty_+ (1) = \{ \underline{a} \in C^\infty_+ \,:\,  \, \sum_{i=1} ^\infty a_i =1 \} \,, 
\end{equation}
\begin{equation}\label{eq-def4}
C^\infty (<\infty ) = \{ \underline{a}\in C^\infty_+ \,:\,  \, \sum_{i=1} ^\infty a_i <\infty  \} \,.
\end{equation}

\subsection{ The  renormalized  von Neumann  entropy}

The most useful local  invariants and  local monotone quantities characterising in  the  qualitative as well as quantitative way quantum correlations, as entanglement of    states in the finite  dimensional  systems, are  defined by means of the  special versions  of   the   entropy  measures, \cf.~\cite{NC, BZ, Guhne2008, 6, HorodeckiReview, AlickiLendi2007}. The  von  Neumann  quantum entropy measure is, without a doubt,  the most  common tool for  these  purposes.

%Let us remind that the inequalities $\|A\|_2\leq \|A \|_1$, where $A\in L_2 (\mathcal H )$ and
%$\|A\| \leq \|A \|_1$, where $A\in L_2 (\mathcal H )$ imply the following continuous embeddings $L_1 (\mathcal H) \subseteq L_2 (\mathcal H)$ and... \footnote{In fact, the  Banach  algebra $L_1(\mathcal H )$  is  a  two-sided  $\ast$-ideal in the  $C^\ast$-algebra $L_\infty (\mathcal H)$.} \Hm{strict inclusion, ????}

Suppose that $\underline{a} \in C^\infty $ and  $a_i \neq 0 $ for all $i$. Moreover we assume that that the limit $
\lim _{n\to \infty }\prod _{i=1}^n a_i$ exists and it is nonzero. Then we say that  the product $\prod _{i=1}^\infty  a_i$ exists. %\Hm{0 ???}
\\
The continuity of $x\mapsto \log x $ implies the following statement.
%\Hm{positivity??}

\begin{lemma}\label{le-conv}
Let $\underline{a} \in C^\infty (< \infty )$.   Then  the product  $\prod _{i=1}^\infty  (1+a_i)$ exists iff~
 $\sum _{i=1}^\infty \log (1+a_i)<\infty $.
\end{lemma}

\begin{lemma}
Let $\underline{a} \in C_+^\infty (1 )$.   Then  the product  $\prod _{i=1}^\infty  a_i$ exists iff~
 $\sum _{i=1}^\infty  (1+a_i)  \log (1+a_i)<\infty $.
\end{lemma}
\begin{proof} The claim follows directly from
\begin{equation}
\log (1+a_i)\leq (1+a_i) \log (1+a_i)\leq  2 \log (1+a_i)\,.
\end{equation}
\end{proof}

Let $A$ be a compact operator in a separable Hilbert space $\mathcal H$ and $\underline{\sigma (A)}$ stands for the discrete eigenvalues of $A$ counted with multiplicities and ordered into  non-increasing sequence. On the other hand, let $\underline{\lambda (A)}$ denote singular values of $A$ counted with multiplicities and forming  non-increasing sequence. If $A\in L_1 (\mathcal H )$ then $\lambda (A)= \sigma (|A^\dag A|^{1/2 })$ and $\sum_{n=1}^\infty\lambda_n <\infty $. The Fredholm determinant takes the form %\Hm{$\det A$ ?}
\begin{equation}\label{eq-determinant}
  \det (\mathrm{I}+A)= \prod_{x\in \lambda (A)} (1+x)\,.
\end{equation}
Below  we remind the basic properties of the Fredholm determinants, \cf.~\cite{G56, S05}.
\\ \\
\begin{theorem}\label{th-main1} [\cite{G56, S05}] \label{th-G56}Let $\mathcal H$ be  a separable  Hilbert  space. Then
\begin{itemize}
\item[\emph i)] For any $\Delta \in L_1 (\mathcal H )$ the map
\begin{equation}
  \mathbb{C} \ni z\mapsto \det (\mathrm I +z \Delta )
\end{equation}
  extends to the entire function which obeys the bound
\begin{equation}
  |\det (\mathrm{I}+ z \Delta ) |\leq \exp (|z| \|\Delta \|_1 )\,.
  \label{lbl:eq:217}
\end{equation}
  \item[\emph ii)] For any maps $L_1 (\mathcal H )\ni \Delta \mapsto \det (\mathrm{I}+ \Delta )$ and $L_1 (\mathcal H )\ni \Delta' \mapsto \det (\mathrm{I}+ \Delta' )$ the following asymptotics
      \begin{equation}\label{eq-Lip}
        |\det (\mathrm{I}+ \Delta )-\det (\mathrm{I}+ \Delta ' )|\leq \|\Delta -\Delta' \|_1 \exp (\mathcal O (\|\Delta\|_1\cdot \|\Delta'\|_1))\,;
      \end{equation}
      in particular $\det $ is the Lipschitz continuous.
  \item[\emph iii)] The following three equivalences hold:
  \begin{equation}\label{eq-detA}
    \det (\mathrm{I}+ z\Delta ) =\exp \left(z \mathrm{Tr}\, [ \log (\mathrm{I}+ z\Delta ] \right)
  \end{equation}
  and
   \begin{equation}\label{eq-detB}
    \det (\mathrm{I}+ z\Delta ) =\sum _{n=1}^\infty  z^n\, \mathrm{Tr}\,[\wedge^n (\Delta )]\,,
  \end{equation}
  where $\wedge^n (\Delta )$ stands for  the antisymmetric tensor power of $\Delta$, see Appendix~\ref{lbl:app:sec:fock:spaces} for more details, and
  \begin{equation}\label{eq-detC}
    \det (\mathrm{I}+ z\Delta ) =\exp  \left( \sum_{n=1}^\infty  \frac{(-1)^{n+1}}{n} z^n \, \mathrm{Tr} \,[\Delta ^n]\right) \,.
  \end{equation}
  %definitions are equivalent are equivalent
\end{itemize}
\end{theorem}
\begin{remarks}
  \rm{ The last equivalence Eq.~\ref{eq-detC} determines so  called  Pelmelj  expansion  with $|z| <1$. For  larger  values of  $|z|$ the   analytic  continuations  are necessary to be  performed.
  	\\
   In the  Appendix~\ref{lbl:app:sec:fock:spaces} we  outline  the methods  of  infinite  dimensional  Grassmann  algebras (the  Fermionic  Fock spaces in  the  physical  notations) as  introduced in the  fundamental  Grothendick memoir~\cite{G56}.
 }
\end{remarks}
In the further discussion we will use the following quantity. %\Hm{only discrete}
%\Hm{skipped index $H$}
\begin{definition} Assume that $Q\in E(\mathcal H )$ and its spectrum $\sigma (Q)=(\lambda_1, \lambda_2,...)$. We define
%\Hm{Maybe we can write generally that $\det A= \prod \lambda_k$}
\begin{equation}\label{eq-FEN}
\mathrm{FEN}_{\pm} (Q)= \log \left(\det (\mathrm{I}+Q)^{\pm (\mathrm{I}+ Q ) } \right)\,,
\end{equation}
where
\begin{equation}\label{eq-detQQ}
\det ( \mathrm{I}+ Q  )^{\pm (\mathrm{I}+ Q ) }= \prod_{k=1}^\infty (1+\lambda_k)^{\pm (1+\lambda_k)}\,.
\end{equation} This means that
\begin{equation}\label{eq-FENlog}
  \mathrm{FEN}_{\pm} (Q)= \pm \sum_{k=1}^\infty (1+\lambda_k ) \log (1+\lambda_k )\,.
\end{equation}
\end{definition} 
In order to relate the above definition with the results formulated in~Thm.~\ref{th-G56} we  introduce  the  following  entropy-generating  operators  $S_\pm$.

\begin{definition} \label{def-S}
For  $Q\in E(\mathcal H )$  we  define
\begin{equation}\label{eq-Spm}
S_\pm  (Q)=(\mathrm{I}+Q)^{\pm (\mathrm{I}+Q)} -\mathrm{1}_{\mathcal{H}}\,,
\end{equation}
where $\mathrm{I}$ means  the  unit  operator here  in the space $\mathcal H$ and spectral functional calculus has been used.
\end{definition}

\begin{remark}
\rm{
In  the  standard, finite dimensional  situation \cite{OP, WGeneral, PQuantum}, the  corresponding  entropic  operator  $S_- (Q)$, formally  looks  like (informally) as
\begin{equation}\label{eq-S-}
S_-(Q) = Q^{-Q} -\mathrm{1}_\mathcal{H}\,.
\end{equation}  }
\end{remark}

Our  definition (\ref{eq-Spm})  is  the  renormalized  (due  to  the  infinite  dimension of  the  corresponding  spaces) version of  it is :  $”(1+Q)^{-(1+Q)}-1”$.

One of the main results reporting on in this note is the following theorem.

\begin{theorem} \label{th-cont}
%The  function(s)  $\mathrm{FEN}_\pm$ defined by~(\ref{eq-FEN}) are  well defined on $E(\mathcal H )j$.
%This means that
For  any  $Q \in E( \mathcal H)$, $\mathrm{FEN}_\pm (Q)$  are  finite and, moreover,  $\mathrm{FEN}_\pm $ are  $L_1 (\mathcal H)$– continuous on $E(\mathcal H)$.
\end{theorem}

The  proof  is  based on  the  following  sequence  of  Lemmas.

Let us define the scalar function
\begin{equation}
f_\pm (x)=(1+x)^{\pm (1+x) }  -1  \,,\,\,\,\,\mathrm{for }\,\,\,\,x\in [0,1]\,.
\end{equation}
%f±(x=(1+x)±1+x−1  on  the  interval  0,1.

\begin{lemma} \label{le-mon}
\begin{itemize}
  \item[\emph{i)}] The  function $f_+ (x)$  is  monotonously  increasing and  convex on $[0,1]$ and
  \begin{equation}
      \begin{array}{lcl}
      % \nonumber to remove numbering (before each equation)
        \inf f_+ (x)   & = & 0 \,,\,\, \mathrm{for}\,\,\,  x=0, \\ 
        \sup f_+ (x) & = & 3 \,,\,\, \mathrm{for}\,\,\, x=1 \,.
      \end{array}
  \end{equation}
  \item[\emph{ii)}]
The  function $f_- (x)$ is  monotonously  decreasing and  concave on $[0, 1]$, and
\begin{equation}
\begin{array}{lcl}
      % \nonumber to remove numbering (before each equation)
        \inf f_- (x)   & = & -0,75 \,,\,\,\, \mathrm{for}\,\,\,  x=1\\ 
        \sup f_- (x) & = & 0 \,,\,\,\,\,\,\quad \quad \mathrm{for}\,\,\, x=0\,.
      \end{array}
 \end{equation}
\end{itemize}
\end{lemma}

   \begin{lemma}\label{le-L1} For any $Q\in E (\mathcal H )$, $S_+(Q) \geq 0 $
and $S_+ (Q) \in L_1 (\mathcal H )$.
\end{lemma}
\begin{proof}
For $0\leq x \leq 1$ the following estimate is valid
\begin{equation}
(1+x)^{1+x}  -1 =\int_0^1 \mathrm{d}s\, \mathrm{e}^{s(1+x)\log (1+x)}(1+x)\log (1+x)\leq
2 \mathrm e ^{2\log 2 }\log (1+x)\,.
\end{equation}
From
\begin{equation}
\mathrm{Tr}\, [(1+Q)^{1+Q} -1] = \sum_{n=0}^\infty ((1+\lambda_n )^{1+\lambda_n }-1)\leq 8
\sum _{n=0}^\infty \log (1+\lambda_n) <\infty\,,
\end{equation}
where we used Th.~\ref{th-G56} and Lemma~\ref{le-conv}.
\end{proof}
%\Hm{both lemmae in one}
 \begin{lemma}\label{le-L1a} For any $Q\in E (\mathcal H )$, $-S_-(Q) \geq 0 $
and $S_- (Q) \in L_1 (\mathcal H )$.
\end{lemma}
\begin{proof}
For $0\leq x \leq 1$ the following estimate is valid
\begin{equation}
-(1+x)^{-(1+x)} + 1 =\int_0^1 \mathrm{d}s\, \mathrm{e}^{-(1-s)(1+x)\log (1+x)}(1+x)\log (1+x)\leq
2 \log (1+x)\,.
\end{equation}
From
\begin{equation}
-\mathrm{Tr}\, [(1+Q)^{1+Q} -1] = - \sum_{n=0}^\infty -( 1 + \lambda_n)^{-(1 + \lambda_n)} + 1\leq 2
\sum _{n=0}^\infty \log (1+\lambda_n) < \infty\,.
\end{equation}

where we have used Th.~\ref{th-G56} and Lemma~\ref{le-conv}.
\end{proof}

In order to prove that the renormalized entropy functions $\mathrm{FEN}_\pm $ are $L_1$ continuous it is enough to prove that the operator valued maps $S_\pm $ are $L_1$ continuous. The latter is proved below.

\begin{lemma} \label{le-L1S}
Let $\mathcal H$ be a separable Hilbert space and let $E (\mathcal H )$ be a space of quantum states on $\mathcal H$. Then the maps
\begin{equation}
Q\mapsto S_\pm (Q ) = (\mathrm{I} + Q )^{\pm (\mathrm{I}  +Q)} -\mathrm{I},
\end{equation}
are $L_1$ continuous on $E(\mathcal H )$.
\end{lemma}
\begin{proof} It is enough to present essential details of the proof for the case $S_+ (Q)$. Let $Q$ and $Q'$ be the states on $\mathcal H$. By the application of the Duhamel formula and equations~(\ref{eq-ineq2.1a}) and~(\ref{eq-ineq2.1b}) we get %\Hm{??}
\begin{equation}
 \begin{array}{l}\label{eq-esmDuhamel}
      % \nonumber to remove numbering (before each equation)
       \|S_+ (Q) - S_+ (Q')\|_1 \leq \sup_{0<s<1} \|\exp s \log (\mathrm{I} + Q) \|  \\
       ~~~~ \cdot   \|\exp (1-s) ( \mathrm{I} +Q')  \log( (\mathrm{I} +Q) \| \\
       ~~~~ ~~~~ \bigg( 
      	 \| (\mathrm{I} + Q) \| \cdot \| \log( \mathrm{I} +Q)- \log( \mathrm{I} +Q') \|_1 \cdot \|\log( \mathrm{I} +Q)\| \cdot  \|Q - Q' \|_1 
       \bigg)\,.
 \end{array}
\end{equation}   
To complete the proof it  suffices to  prove the norm continuity  of  the  operator valued function  $\log( \mathrm{I} +Q)$.
Let $Q\in E(\mathcal{H})$. Define $\tau (Q) = \sup \sigma (Q)$. Then $\tau (Q)\leq 1$ and $\|\log( \mathrm{I} +Q) \| = \log( 1+\tau(Q))$. Let $Q\,,Q' \in E(\mathcal{H})$ with $\|Q-Q'\|_1\leq \delta <1$. Using again formulae~(\ref{eq-ineq2.1a}) and~(\ref{eq-ineq2.1b})  we have %\Hm{Sprawdzic}
\begin{equation}
\begin{array}{lcl} \label{eq-ineq1}
\| \log( \mathrm{I} +Q) - \log( \mathrm{I} +Q') \|_1 & \leq & \sum_{n=1}^\infty \frac{1}{n}\|Q^n -(Q')^n\|_1 \\ \\ 
  & \leq & \sum_{n=1}^\infty \frac{1}{n}  \sum_{k=1}^n   \|Q^{k-1} (Q-Q' )(Q')^{n-k}\|_1  \\ \\
  & \leq &\sum_{n=1}^\infty \frac{1}{n}  \sum_{k=1}^n   \tau (Q)^{k} \tau (Q' )^{n-k-1}\|Q-Q'\|_1\,.
\end{array}
\end{equation}
Let $\tau: = \max \{ \tau (Q)\,,\tau (Q') \}<1 $. Then summarizing the above reasoning we have
\begin{equation}\label{eq-in-log}
\| \log  ( \mathrm{I} +Q) - \log  ( \mathrm{I} +Q')\|_1\leq \frac{\delta }{1-\tau }\,.
\end{equation}
The analysis of further properties, together with the analysis of the case $\tau=1$ of the map $Q\mapsto \log  ( \mathrm{I} +Q)$ we postpone to Appendix~\ref{lbl:app:operator:renorm:map}.
\end{proof}

\newcommand{\mpreccurlyeq}{\;{\rm\scriptsize m}\mbox{-}\mskip-5mu\preccurlyeq}
\newcommand{\gpreccurlyeq}{\;{\rm\scriptsize g}\mbox{-}\mskip-5mu\preccurlyeq}

\begin{proposition} %\Hm{$\mathrm I -> \mathrm I_\mathcal H$}
Let  $\mathcal H$  be  a  separable Hilbert  space and  let $E(\mathcal H )$  be  the  space  of  quantum states  on  $\mathcal H$.
Then  the $L_{\infty}$  norms (spectral  norm) of the  entropy  maps
$S_{\pm}(Q)=(\mathrm I+Q)^{\pm(1+Q)}-\mathrm I$ are  given  by:
\begin{enumerate}
\item $\| S_{+}(Q)\|_{\infty}=(1+\tau_1)^{1+\tau_1}-1$,\quad  \rm{where} $\tau_1 =\sup (\sigma ( Q))$,
\item $\sup_{Q\in E(\mathcal H)} \|S_+(Q)\|_{\infty}=3$,
\item $\| S_{-}(Q)\|_{\infty}=1-(1+\tau_1)^{-(1+\tau_1)}$, \quad \rm{where} $\tau_1 =\sup (\sigma ( Q))$,
\item $\sup_{Q\in E(\mathcal H)} \|S_-(Q) \|_{\infty} = 0.75$.
\end{enumerate}
\end{proposition}
\begin{proof}
%\begin{remarks}
For the compact operators it is known that $\|Q\|_\infty = \|Q\|$, see~\cite{RSI}.
\end{proof}
\begin{remark}%Remarks.

\rm{ Let  us  assume  that  $\dim(\mathcal{H}) =d$  and  is  finite. Then, taking a pure  state  $Q$, i.e. the  state  for  which
$\mathrm{Tr} [Q] = \mathrm{Tr}[ Q^2]=1$  it  follows  that  the  value  of  renormalized  entropy $\mathrm{FEN}_+ (Q)$ of $Q$
which  has  the  rank of  Schmidt  equal  to one, is  equal  to
$2\log(2)$  ($\mathrm{FEN}_{–}(Q)=-2\log(2)$) and  is  independent  of  $d$. Taking  maximally  mixed   state $Q$
with  the  spectral numbers $\sigma (Q) = (1/{ d}, \dots 1/{d} )$ we  have
$\mathrm{FEN}_+(Q)= (1+ 1/d)\log( 1+1/d)^{d} ~ (\mathrm{FEN}_{-}(Q)= -(1+ 1/d)\log( 1+1/d)^{-d} )$
which tends monotonously, as  $d$  tends  to infinity  to  the  value $1$  (resp.  to  the  value $–1$).}
\end{remark}
The use of   standard  not  renormalized definition of  entropy  of  entanglement leads  to the  statement  that it  is
taking  values in  interval $[0, \log(d)]$, which  shows  that  there  is  no possible  straightforward  passage  from  the
finite  dimensional  situation to  the  infinite  dimensional  systems. The widely used, another  entropic measures  of
entanglement \cite{NC, BZ, Guhne2008, HorodeckiReview} also  must be  suitable  renormalized in  order to be  applied  in infinite  dimensions
in a  way that  overcome   the  several  discontinuity  and divergences problems as  well  problems  arising in  the  genuine
infinite dimensional cases. The  results on  this  will  be  presented in  a separate  note. %\Hm{czy tu koniec remark}

Let  $Q (n)$  be  the  sequence  of  $L_1 ( \mathcal H )$ such that the  $n$-th  first eigenvalues  of $Q(n)$  is  equal  to  $1/n$  and  the  rest  of  spectrum  is  equal  to  zero. The  renormalized  entropy  of  $Q(n)$ is given by  \begin{equation}\label{eq-FEN+}
\mathrm{FEN}_{+}( Q(n))= (n+1) \log\left(1+ \frac{1}{n}\right)\,.
\end{equation}
It  is  easy  to  see  that $\lim_{n \rightarrow \infty} \mathrm{FEN}_{+} (Q(n)) = 1$.

\begin{theorem}
For  any  sequence of  states  $Q (n) \in L_1(\mathcal H)$  as  above there  exists   state $Q^*$ in  $E(\mathcal{H})$ and  such
$\mathrm{FEN}_+ (Q^* )=1$.
\end{theorem}
\begin{proof}
For  any such  sequence $Q(n)$  we  apply the Banach--Alaglou theorem  first, concluding that  the  set $ \{Q(n)\}$ forms
$^\ast $-week precompact  set  and  therefore,  in the  $^\ast $-week topology  $\lim Q(n)$ by subsequences  do  exists.
However, these  limits  are  all  equal  to  zero. In  order  to  obtain   a non-trivial  result  we  use  the  Banach--Saks
theorem which  tells  us  that  there  exists  a  subsequence $n_j$  such  that  the  following  Cesaro  sum  of  $Q(n)$
\begin{equation}\label{eq-Cesaro}
 C_M(Q)=\frac{1}{M}\sum_{j=1}^M Q(n_j)
\end{equation}
is  strongly  convergent as  $M\rightarrow \infty$   to  some non-zero  operator $Q^\ast \in E (\mathcal{H})$.
\end{proof}

%\begin{remark}
It  would  be interesting to  describe  in the  explicit way  the most  mixed  states  i.e.  the  states  for
which the  value  of  $\mathrm{FEN}_+ (Q) =1$.
%\end{remark}

\section{Some remarks on the majorization  theory} \label{lbl:sec:majorisation:theory}

The  fundamental  results  obtained in Alberti  and  Uhlmann  monograph \cite{27} and  applied  so  fruitfully to  the  quantum
information  theory by  many researchers (see  \cite{NC, BZ, 27, 28, 32, 20}  and  references  therein, are   known widely today  under  the
name  (S)LOCC majorization  theory (in the  context  of quantum information theory). Presently this theory  is  pretty  good
understood in the  context  of  bipartite,  finite  dimensional  systems, (especially  in  the  context  of  pure states), see \cite{NC, BZ, 28}. In  papers \cite{29, 31, 32, 33, Kaftal2010} successful attempts  are  presented in  order  to  extend  this  theory  to
the  case  of   bipartite ,infinite dimensional systems. Below, we  present  some  remarks  which  seems  to be useful in
this  context.

For  a  given $\underline{a} \in C^\infty$   we  apply  the  operation of  ordering  in  non-increasing  order and
denote the  result  as $\underline{a}^{\geq}$. Of  particular  interest  will be  the  image of this operation,
when applied   pointwise to  the  infinite  dimensional  simplex
$C^{\infty}_{+}(1) := \{\underline{a}=(a_n) \in \R^N, a_n \geq 0,
\sum_{i=1}^{\infty} a_i=1\}$. This will be denoted as $C^{\geq}$.
Let  us  recall  some standard   definitions   of  majorisation theory. Let $\underline{a}, \underline{b} \in C^{\geq}$. Then  we  will  say that $\underline{b}$ is majorising $\underline{a}$ iff  for  any  $n$ the following is satisfied
\begin{equation}\label{eq-major}
\sum_{i=1}^{n} a_i \leq \sum_{i=1}^n b_i\,.
\end{equation}
If  this  holds  then  we  denote  this  as $\underline{a} \preccurlyeq  \underline{b}$.

We  will say  that $\underline{b}$  is majorizes multiplicatively $\underline{a}$ iff  for  any  $n$ the following is satisfied
\begin{equation}\label{eq-major-multip}
 \prod_{i=1}^n (a_i +1) \leq \prod_{i=1}^n (b_i +1)\,.
\end{equation}
If  this  holds to be true then  we  denote  this fact  as $\underline{a} \; \mpreccurlyeq \underline{b}$.

Let  $F$  be  any  function (continuous, but  not necessarily) on  the  interval $[0, 1]$. The  action of $F$ on
$C_{+}^{\infty}$ (and  other  spaces  of  sequences  that  do  appear) will  be  defined
$ (F(a_i))$.

Recall the  well  known  result, see  i.e.  [5,6].

\begin{lemma}\label{le-log}
Let  as  assume  that $f$  is  a continuous,  increasing and  convex  function  on $\R$.
If $\underline{a} \preccurlyeq  \underline{b}$ then $f(\underline{a}) \preccurlyeq  f(\underline{b})$.
\end{lemma}

It  is  clear  from the very  definitions that $\underline{a} \mpreccurlyeq \underline{b}$   iff
$\log(\underline{a}+1) \preccurlyeq  \log(\underline{b}+1)$.

\begin{proposition}
Let $\underline{a}, \underline{b} \in C^{\geq}$ and  let as  assume  that
$\underline{a} \mpreccurlyeq \underline{b}$.
Let  $f$  be  continuous, increasing  function and such  that  the  composition  $f\circ \exp (x)$
is  convex  on  a  suitable  domain. Then $f(\underline{a}) \preccurlyeq  f(\underline{b})$.
\end{proposition}
\begin{proof}
For fix  $n$ we have:
\begin{equation}\label{eq-prod1}
 \prod_{i=1}^n (a_i +1) \leq \prod_{i=1}^n (b_i +1)\,.
\end{equation}
Taking  $\log$ of both  side  we  obtain
\begin{equation}\label{eq-prod2} \sum_{i=1}^n \log(a_i +1) \leq \sum_{i=1}^n \log(b_i +1)\,.
\end{equation}
Applying  Lemma~\ref{le-log}  we  obtain
\begin{equation}\label{eq-prod3} \sum_{i=1}^n f (a_i) \leq \sum_{i=1}^n f(b_i)\,.
\end{equation}
\end{proof}
%\circ \mathrm{exp}

In  particular  taking  $f(x) =x$  we  conclude
\begin{corollary} \label{cor-seq-a}
Let  and  let as  assume  that
$\underline{a}, \underline{b} \in C^{\geq}$ and $\underline{a} \mpreccurlyeq \underline{b}$.
Then $\underline{a} \preccurlyeq \underline{b}$.
\end{corollary}

The last result says  that  each  linear  chain  of  the  semi-order  relation $\mpreccurlyeq$
in $C^\geq $  is  contained  in  some  linear  chain  of  the  semi-order  $\preccurlyeq$.
It means  that  the  semi-order $\mpreccurlyeq$ is  finer  than  those  induced  by $\preccurlyeq$.

%\Hm{STOP HERE}
\begin{corollary} \label{cor-seq}
Any  $\preccurlyeq$-maximal element in $C^\geq$ is  also  $\mpreccurlyeq$-maximal.
\end{corollary}
\begin{proof}
If  $\underline{a} \mpreccurlyeq \underline{b} $  then $\underline{a}  \preccurlyeq \underline{b} $.
Let  $\underline{a}^\ast $  be  a $\preccurlyeq$-maximal in $C^\geq $ and let us assume that there
exists $\underline{a}^{\ast \ast }$ such that $\underline{a}^{\ast } \mpreccurlyeq \underline{a}^{\ast \ast } $ and  the  contradiction  is  present.
\end{proof}

To complete   this  subsection  we  quote  the  infinite  dimensional  extension of  the  majorisation  theory  applications in
the  context  of  quantum information  theory.

For  this  goal  let  us  consider  any  $Q \in  E(\mathcal{H})$, where $\mathcal{H}$ is  a  separable
Hilbert  space. With any  such  $Q$  we  connect  a  sequence  $(P_{sp}(N))$ of  finite  dimensional projections
$P_{sp}(Q)$ which  we  will  call  the    spectral  sequence of $Q$. This is  defined  in the  following  way:
let $Q = \sum_{n=1}^{\infty} \tau_n E_{\phi_n}$  be  the  spectral  decomposition  of  $Q$  rewritten  in  such
a way  that eigenvalues $\tau_n$ of $Q$ are  written in  nonincreasing  order. Then  we  define
$P_{sp}(Q) (n)=\oplus_{i=1}^n E_{\phi_n}$.
Finally, we  define  a  sequence  of  Gram  numbers $g_n(Q)$ connected  to  $Q$:
\begin{equation}\label{eq-Gram}
\underline{g(Q_1)} = (g_n(Q)=\det(\mathrm{I} + QP_{sp}(Q)(n)) )\,.
\end{equation}
%%%%%%%%%%%%%%%%%%

\begin{definition} \label{def-Gram}
Let $Q_1, Q_2 \in E(\mathcal{H})$.
We will say the  $Q_2$ m-majorises $Q_1$  iff  $g_n (Q_1) \leq g_n (Q_2)$ for all $n$. This  will be  written  as
$Q_1 \mpreccurlyeq Q_2$.
\end{definition}

Let $Q_1 , Q_2 \in E(\mathcal H )$. The standard definition of majorisation is the following: $Q_2$ majorizes $Q_1$ iff $\underline{\sigma(Q_1) }\preccurlyeq \underline{\sigma(Q_2)}$.

\begin{proposition}
Let  $\mathcal{H}$ be  separable  Hilbert  space  and  let $Q_1, Q_2 \in  E(\mathcal{H})$  be  such that
$S_{+}(Q_1) \mpreccurlyeq S_{+}(Q_2)$. Then
\begin{enumerate}
\item  $\mathrm{FEN}_{+} (Q_1) \leq  \mathrm{FEN}_{+} (Q_2)$,
\item  $\mathrm{FEN}_{-} (Q_1) \geq  \mathrm{FEN}_{-}( Q_2)$,
\item  $\mathrm{FEN}_{+} (Q_1) = \mathrm{FEN}_{+}(Q_2)$  iff     $\sigma(Q_1) = \sigma (Q_2)$,
\item  $\mathrm{FEN}_{-} (Q_1) = \mathrm{FEN}_{-}(Q_2)$  iff     $\sigma(Q_1) = \sigma (Q_2)$\,.
\end{enumerate}
\end{proposition}

\begin{proof}
The  point  (1)  and  (2)  follows  from the  fact  that majorisation  in  the  sense of  Definition~\ref{def-Gram}
is  equivalent to  the m-majorisation  of the  considered  entropy generating  operators from  which follows,
using Corollary~\ref{cor-seq-a}, that   they are  also in  the  standard majorisation  relation.

More  details  for  this: let  $\sigma(Q_1) = (\lambda_n)$ and $\sigma(Q_2) = (\mu_n)$. Then
$\sigma( S_{+}(Q_1))=((1+\lambda_k)^{1+\lambda_k} -1)$ and
similarly for $\sigma(S_{+}(Q_2))=( (1+\mu_k)^{1+\mu_k} -1)$.
From  Corollary 3.12  it  follows:
\begin{equation}\label{eq-major1a}
((1+\lambda_k)^{1+\lambda_k} )\preccurlyeq ((1+ \mu_k)^{1+\mu_k})\,.
\end{equation}
Using the  fact  that  log  is  convex it  follows  that
\begin{equation}\label{eq-major1}
((1+\lambda_k)\log(1+\lambda_k)) \preccurlyeq ((1+\mu_k)\log(1+\mu_k))\,.
\end{equation}
Application  the standard, finite  dimensional arguments  leads  to  the  inequalities:
$$ \mathrm{FEN}_{+} (Q_1 P_{sp} (n)) \leq  \mathrm{FEN}_{+} ( Q_2 P_{sp})(n)).$$
Using  the  $L_1$ convergence  $\lim_{n \rightarrow \infty} P_{sp}(Q)(n)=Q$
and the continuity of $\mathrm{FEN}_\pm $ the proof of (1) follows.
The  proof  of (2) is almost  identical  to that for~(1).

To  prove~(3) and~(4)  let us introduce  the  following interpolation:
if  $\sum_{n=1}^{\infty}\lambda_n E_{\phi_n}$, resp. $\sum_{n=1}^{\infty}\mu_n E_{\omega_n}$
are the spectral decompositions of $Q_1$, resp. $Q_2$ then
\begin{equation}\label{eq-Q}
Q(t) = \sum_{n=1}^{\infty}(t \lambda_n + (1-t) \mu_n) E_{\phi_n}\,.
\end{equation}
It is  easy to  see that assuming  $Q_1 \preccurlyeq  Q_2$
%% czerwone
\begin{equation}
\sigma(Q_1) \preccurlyeq \sigma(Q(t)) \preccurlyeq \sigma(Q_2)
\end{equation}
from  which  we  conclude  that  if   $\mathrm{FEN}(Q_1) =\mathrm{FEN}(Q_2)$  then $\mathrm{FEN}(Q(t)) =\mathrm{const}$.
It  is not  difficult   to prove  that
\begin{equation}
\mathrm{FEN} (Q(t)) =\sum_{n=1}^{\infty}(1+t\lambda_n+(1-t)\mu_n)\log(1+t\lambda_n+(1-t)\mu_n)
\end{equation}
as  function  of  $t$  is  smooth. Calculating  the  second  derivative  of  its  we  find
\begin{equation}
\frac{d^2}{dt^2}\mathrm{FEN} (Q(t))= \sum_{n=1}^{\infty} \frac{ (\lambda_n-\mu_n)^2}{1+t\lambda_n+(1-t)\mu_n} =0\,.
\end{equation}
This completes the proof.
\end{proof}

Before we present (after \cite{27, 29, 31} and with  minor  modifications) infinite dimensional generalisation of the fundamental in this context Alberti-Uhlmann theorem we briefly recall some  definitions.

A  completely  positive  map $\Phi$ on  a  von Neumann  algebra  $L_{\infty}(\mathcal{H})$   is said to be normal if
$\Phi$ is continuous with respect to the ultraweak ($\ast$-weak) topology. Normal completely positive contractive maps on
$B(\mathcal{H})$ are characterized by the theorem of Kraus which says that $\Phi$ is a
normal completely positive map if and only if there exists at  least  one  sequence
$(A_i)_{i=1,\dots}$ of bounded operators in $L_{\infty}(\mathcal{H})$ such that for   any
$Q \in L_{\infty}(\mathcal{H})$:
\begin{equation}\label{eq-Kraus}
\Phi(Q)=\sum_{i=1}^{\infty} A_i Q A_i^{\dagger}\,,
\end{equation}
where
\begin{equation}
\sum_{i=1}^{\infty} A_i A_i^{\dagger} \leq \mathrm I_{\mathcal H},
\end{equation}
and where the limits are defined in the strong operator topology.
A normal completely positive map $\Phi$ which is trace preserving is  called a  quantum  channel.
If  a normal completely positive map $\Phi$ satisfies  $\Phi( \mathrm I_{\mathcal H} ) \leq \mathrm I_{\mathcal H}$
then called a quantum operation. A quantum operation $\Phi$  is  called  unital iff
$\Phi(\mathrm I_{\mathcal H}) =\mathrm I_{\mathcal H}$ which is equivalent to $\sum_{i=1}^{\infty} A_iA_i^{\dagger}=\mathrm I_{\mathcal H}$ for  some Krauss decomposition of $\Phi$.

A  quantum operation $\Phi$ is called  bistochastic  operation if it is both trace preserving and  unital. Central  notion for us is the notion of a mixed unitary operation.

A quantum  operation $\Phi$ is called a (finite) mixed unitary  operation iff there  exists a  (finite)  ensemble
$\{ U_i \}_{i=1:n}$ of  unitary  operators on $\mathcal{H}$ and  a (finite)  sequence  $p_i \in [0,1]$ such that
$\sum_{i=1}^n p_1=1$ and
\begin{equation}\label{eq-U}
\Phi(Q)=\sum_{i=1}^n p_iU_iQU_i^{\dagger}\,.
\end{equation}
\begin{theorem} %\Hm{wrocic} %\Hm{??}
Let $\mathcal{H}$ be  a separable   Hilbert  space   and  let  $Q_1, Q_2 \in E(\mathcal{H})$. Assume  that
$Q_1 \mpreccurlyeq  Q_2 $. Then there exists  a sequence  $(\Phi_n)$ of mixed unitary operations
and  a limiting bi-stochastic operation $\Phi^\ast$ such  that 
%for  any $Q \in E(\mathcal{H})$
the  sequence   of  states $\Phi_n(Q_2)$ is $L_1$-convergent  to  $\Phi^\ast (Q_2)=Q_1$.
\end{theorem}

\begin{proof}
The only essential difference comparing to the original formulation of this result \cite{27, 29, 31} is that instead of $\preccurlyeq$ type majorisation $\mpreccurlyeq $ is used. 
\end{proof}

Also  the  following  result holds

\begin{theorem} \label{lbl:thm:38}
Let  $\mathcal{H}$  be  a separable  Hilbert  space  and  let $\Phi$
be any quantum operation acting on $E(\mathcal{H})$. Then
\begin{equation}
\det(\mathrm I_{\mathcal H}+\Phi(Q))\leq \det(\mathrm I_{\mathcal H}+Q)\,.
\end{equation}
\end{theorem}

\begin{proof}
Let $T$ be any non-expansive linear operator acting  on  $\mathcal{H}$ -- this means that the operator norm of $T$, $\|T\|\leq 1$. Using  the Grothendick formula (\ref{th-Agrod})  and the  following  reasoning:
\begin{equation}
\begin{array}{lcl} %\label{eq-1} 
\mathrm{Tr} \, [ \wedge^n (TQT^{\dagger}) ] & = & \sum_{i_{1_<\dots <i_n}} \langle i_1 \dots i_n |\,(TQT^{\dagger}) ^{\otimes n}\, | i_1 \dots i_n \rangle
\\  \\
& = & \sum_{i_{1_<\dots <i_n}} \prod_{k=1:n} \langle i_k |\,TQT^{\dagger}\,| i_k \rangle  \\ \\ %
%&\leq & \sum_{1_<\dots <i_n} \prod_{k=1:n} <i_k |TQT^{\dagger}| i_k> %
& \leq & \sum_{i_{1_<\dots <i_n}} \langle i_1 \dots i_n |\,Q ^{\otimes n}\,
|i_1 \dots i_n \rangle = \mathrm{Tr}[\wedge^n Q ]\,,
\end{array}
\end{equation}
where  we have  used  the  assumption  that the  norm $T$  of  is  not  bigger then 1 and  positivity  of $Q$.

Now, let  us  assume that  we  have  a  pair  of  bounded  operators  $T_1,  T_2$ and  such that
$T_1T_1^{\dagger} + T_2T_2^{\dagger} \leq \mathrm I_\mathcal H$. For $Q \in E(\mathcal{H})$:
\begin{equation}
\begin{array}{lcl}
\mathrm{Tr}[ \wedge^n(T_1QT_1^{\dagger} + T_2QT_2^{\dagger}) ]  & = & \sum_{i_1< \dots < i_n} \prod_{k=1:n} \langle i_k | (T_1T_1^{\dagger}+T_2T_2^{\dagger})Q|i_k \rangle \\ \\
& \leq & \sum_{i_1  < \dots <i_n}  \langle i_1 \dots i_n |Q^{\otimes n}|i_1 \dots i_n \rangle \\ \\
& = & \mathrm{Tr} [\wedge^n(Q)].
\end{array}
\end{equation}
Now, the general case follows by application of Krauss representation theorem for quantum  operations~(\ref{eq-Kraus}) and some elementary inductive and  continuity  arguments.
\end{proof}

Several additional  results on renormalised version of von Neumann entropy, in particular on the invariance and monotonicity properties of von Neumann entropy  in the  infinite  dimensional setting of conditional  entropies  are included in \cite{GierelarkFuture2}.

\section{The case of tensor product of states} \label{lbl:sec:tensor:states}

\subsection{Renormalized Kronecker products}

Let us  recall  the  finite  dimensional formula  for  computing  determinant  of  tensor  product  of  matrices.

\begin{lemma}[Kronecker  formula]
Let  $\mathcal{H}_A$  and  $\mathcal{H}_B$  be  a  pair  of  finite  dimensional  Hilbert  spaces  with  dimension
$N_A$, and  resp. $N_B$. Then, for any  $Q_A \in  L(\mathcal{H}_A)$  and   $Q_B \in  L(\mathcal{H}_B)$
the  following  formula  is  valid
\begin{equation}\label{eq-Kron}
    \det ( Q_A \otimes Q_B  ) = (\det( \mathrm I_A\otimes  Q_B))^{N_A} \cdot (\det(Q_A\otimes \mathrm I_B))^{N_B}\,.
    \end{equation}
\end{lemma}

\begin{proof}(quick-argument  based).
Let stands $\mathrm{I}_A$, respectively $\mathrm{I}_B$ stands for  the  unit  operators  in  the  corresponding  spaces $\mathcal{H}$. Then
\begin{equation}
	Q_A \otimes Q_B  = ( \mathrm{I}_A \otimes Q_B )( Q_A \otimes \mathrm{I}_B )
\end{equation}
from  which  it  follows  easily the  Kronecker  formula~(\ref{eq-Kron}).
\end{proof}

If  one  of  the  factors in~(\ref{eq-Kron})  is  infinite dimensional and the %\Hm{back}
determinant (absolute  value  of)  of  the  corresponding  matrix  $Q$  is   strictly bigger  then  one
(or strictly  smaller then one) then the value  $\mathrm{det}$  of  the  product(\ref{eq-Kron})  is  infinite, respectively  equal  to  zero.

In order  to understand  better  this  problem  we  define  renormalized  Kronecker  product %\Hm{index r}
\begin{equation}
( \mathrm{I}_A +  Q_A )\otimes_r ( \mathrm{I}_B + Q_B ):= \mathrm{I}_{\mathcal H} + Q_A \otimes Q_B
\end{equation}
which  formally  can be  written  as:
\begin{equation}
( \mathrm{I}_A +  Q_A)\otimes_r (\mathrm{I}_B + Q_B):= (\mathrm{I}_A +Q_A ) \otimes (\mathrm{I}_B +Q_B) - Q_A\otimes \mathrm{I}_B - \mathrm{I}_A\otimes Q_B\,.
\end{equation}

\begin{proposition}
Let $\mathcal{H}_A$ and  $\mathcal{H}_B$ be  a pair  of  separable  Hilbert  spaces  of  an  arbitrary  dimensions
$\mathcal{H}= \mathcal{H}_A \otimes \mathcal{H}_B$ and  let  $Q_A \in E(\mathcal{H}_A)$ and
$Q_B \in E(\mathcal{H}_B)$. Then the map
\begin{equation}
     z \mapsto     \det (\mathrm{I}_{\mathcal H} + zQ_A \otimes Q_B )
\end{equation}
defines an entire function in the complex plane and such that the following  estimate is valid:
\begin{equation}
  |\det (\mathrm{I}_{\mathcal H} + z Q_A \otimes Q_B )| \leq \exp(|z|) .
\end{equation}
\end{proposition}

%\begin{proof}
The  proof  is  an  immediate  consequence  of  the Theorem~\ref{th-main1}~(i) and Le\-mma~(\ref{le-aux2}) below.% f and~Proposition  2.5. %\Hm{Proposition 2.5???}
\begin{lemma}\label{le-aux2}
Let $Q_A \in E(\mathcal{H}_A)$  and  $Q_B \in E(\mathcal{H}_B)$.
Then $Q_A\otimes Q_B  \in E(\mathcal{H}_A \otimes \mathcal{H}_B)$.
\end{lemma}
\begin{proof}
%First  we assume  that both $Q_A$ and $Q_B$ are nonnegative.
Recall  that  the  spectrum  $\sigma(Q_A \otimes Q_B)$ is  given  by
\begin{equation}
  \sigma(Q_A \otimes Q_B) = ( \lambda \mu, \lambda \in \sigma(Q_A), \mu \in \sigma(Q_B))
\end{equation}
from  which it  follows:
\begin{equation}
	\mathrm{Tr} [Q_A \otimes Q_B  ] =  \mathrm{Tr} [ Q_A]\cdot  \mathrm{Tr} [Q_B] =1\,.
\end{equation}	
This completes the proof.
\end{proof}

Another  renormalisation  of  the  tensor  product  can be  achieved  by  the  use  of  infinite  dimensional Grassmann
algebras as  we  have  outlined  in the  Appendix A  to  this  note. For  this  goal  let  us  define
\begin{equation}
(\mathrm{I}_{\mathcal H_A} +  Q_A )\otimes_{fr} (\mathrm{I}_{\mathcal H_B}+Q_B):= (\mathrm{I}_{\mathcal H_A}+ Q_A) \wedge (\mathrm{I}_{\mathcal H_A}+Q_B)\,,
\end{equation}
where $\wedge $   stands  for   skew (antisymmetric)  tensor  product  and  the  right  hand  side  here  is  defined as
a one  particle  operator in  the  skew  Grassmann  algebras  built  on
$\mathcal{H}_A$ and  $\mathcal{H}_B$, see  Appendix A. Using  the  unitary  isomorphism  map  $J$  in  between  the
antisymmetric  product  of  fermionic  Fock  spaces  build  on the  spaces  $\mathcal{H}_A$
and $\mathcal{H}_B$  (see Appendix A  and  the  Theorem~\ref{th-deter}) and  the  antisymmetric  Fock  build  on  the   space
$\mathcal{H}_\oplus = \mathcal{H}_A \oplus \mathcal{H}_B$ we can define
\begin{equation}
\det\bigg(  (\mathrm{I}_A +  Q_A )\otimes_{fr} (\mathrm{I}_B+ Q_B) \bigg) := \det \bigg(\mathrm{I}_{\mathcal H_A \oplus \mathcal H_B} + Q_A\oplus Q_B \bigg)\,.
\end{equation}
%\Hm{STOP HERE}
\begin{theorem}
Let  $\mathcal{H}_A$ and  $\mathcal{H}_B$ be  an pair  of  separable  Hilbert  spaces  of  an  arbitrary  dimensions and
$\mathcal{H}= \mathcal{H}_A\otimes \mathcal{H}_B$ and  let  $Q_A \in L_1(\mathcal{H}_A)$ and
$Q_B \in L_1 (\mathcal{H}_B)$. Then  the   map
\begin{equation}\label{eq-sum2} z \rightarrow   \det (\mathrm{I}_{\mathcal H_A \oplus \mathcal H_B} + z Q_A \oplus Q_B)\end{equation}
defines  an  entire  function   in  the  complex plane and  such  that the  following  estimate  is  valid:
\begin{equation}\label{eq-sum}
|\det (\mathrm{I}_{\mathcal H}+ z Q_A\oplus Q_B )| \leq \exp(|z| (\|Q_A \|_1 + \|Q_B\|_1)\,.
\end{equation}
\end{theorem}
\begin{proof} 
%The  right  hand  side   of  (  ), for  $z=1$  
As  we  have  proved  in  the  Theorem~(\ref{th-deter})  the right hand side of~(\ref{eq-sum2}) is equal to the product $\det (\mathrm{I}_A + Q_A ) \det(\mathrm{I}_B+Q_B )$. Having this the claim of this theorem follows by a straightforward application of Theorem~\ref{th-main1}~(i). follows.
\end{proof}

\begin{remark}
For an interesting paper on the influence of quantum statisics on the entanglement see i.e. \cite{8}.
\end{remark}

Another   interesting  implication of  Theorem~\ref{th-deter}  seems  to be  the  following  observation.%\Hm{seems??}

\begin{theorem}
Let $\mathcal{H}=\oplus_{i=1}^{N} \mathcal{H}_i$ and $Q\in L_1 (\mathcal H )$ and such that $Q= \oplus \lambda_i Q_i $, where $Q_i  \in E (\mathcal{H}_i)$ for all $i=1,...$, $\lambda_i \geq 0 $, $\sum_i ^N \lambda_i =1$.%modulo normalisation
%(which means that the  absence  of  any  interaction in between  different  sort of  particles  is  assumed).
Then $Q\in E (\mathcal H)$ and  
\begin{equation}\label{eq-ineq}
	\mathrm{FEN}_{\pm}(Q) =  \sum_{i=1}^{N} \mathrm{FEN}_{\pm}(\lambda_i Q_i)\,.
\end{equation}
\end{theorem}
\begin{proof}
Let  us  observe  that the  renormalized  entropy  operators  $S_\pm $
can be  decomposed  as:
\begin{equation}
S_{\pm}(Q) = \oplus_{i=1}^{N} S_{\pm}(\lambda_i Q_i) = \oplus_{i=1}^{N} \left[(\mathrm{I}_{\mathcal H_i} + \lambda_i Q_i)^{\pm ( \mathrm{I}_{\mathcal H_i} + \lambda_i Q_i )}- \mathrm{I}_{\mathcal H_i} \right]\,.
\end{equation}
Therefore  using Theorem~\ref{th-deter2}  we  obtain %\Hm{back}
\begin{equation}
\begin{array}{lcl} 
\mathrm{FEN}_{\pm}(Q)  &=& \log \det (\mathrm{I}_{\mathcal H} + S_{\pm}(Q)) \\ \\
 &=& \log \left(\prod_{i=1}^N \det(\mathrm{I}_{\mathcal H_i} + S_{\pm}(\lambda_i Q_i)) \right) \\ \\
 &=& \sum_{i=1}^N \mathrm{FEN}_{\pm} (\lambda_i Q_i).
\end{array}
\end{equation}
\end{proof}

Also  the  following  result  seems  to be  interesting. %\Hm{seems??}
\begin{theorem}
Let $\mathcal{H} = \mathcal{H}^A \otimes \mathcal{H}^B$ be a separable Hilbert space and let $\Phi$ be a separable quantum operation on $\mathcal{H}$, i.e. $\Phi=\Phi^A \otimes \Phi^B$, where $\Phi^A$, resp. $\Phi^B$ are local quantum operations. Then for any $Q \in E(\mathcal{H})$:
\begin{equation}
	\det \big(  1_\mathcal{H} + \Phi(Q) \big) \leq \det ( 1_\mathcal{H} + Q ) .
\end{equation}
\end{theorem}

\begin{proof}
Let $K^A_i$, $i=1,\ldots$, resp. $K^B_j$, $j=1,\ldots$ be the families of operator giving the Kraus representations, for
\begin{equation}
	\Phi^A(A) = \sum_{i=1} K^A_i A {K^A_i}^{\dagger} ,
\end{equation}
and, resp.
\begin{equation}
	\Phi^B(A) = \sum_{i=1} K^B_i A {K^B_i}^{\dagger} .
\end{equation}
Then, for any $Q \in E(\mathcal{H})$:
\begin{equation}
	\Phi(Q) = \sum_{i,j } K^A_i \otimes K^B_j (Q) {K^A_i}^{\dagger} \otimes {K^B_j}^{\dagger} .
\end{equation}
Taking into account that 
\begin{equation}
\begin{array}{lcl}
\sum_{i,j } \bigg( K^A_i \otimes K^B_j \bigg)  \bigg( {K^A_i} \otimes {K^B_j} \bigg)^{\dagger} & = & 
\bigg(\sum_{i} K^A_i \cdot {K^A_i}^{\dagger} \bigg) \otimes \bigg( \sum_{i=1} {K^B_i} \cdot {K^B_i}^{\dagger} \bigg) \\ \\
& \leq & 1_{\mathcal{H}_A} \otimes 1_{\mathcal{H}_B} 
= 1_{\mathcal{H}} ,
\end{array}
\end{equation}
the proof follows as the proof of Theorem~\ref{lbl:thm:38}.
\end{proof}

\subsection{Reduced density matrices -- the bipartite case}

Let $\mathcal H = \mathcal H _A\otimes \mathcal H_B$ be the tensor product of two separable Hilbert spaces
$\mathcal H _A$ and $\mathcal H_B $ of arbitrary dimensions. In this section we assume that
both spaces $\mathcal H_A$, $\mathcal H_B$ are infinite dimensional (everything works also in
finite dimensional situations \cite{GierelarkFuture1}, and also in situation for which only one
of the spaces $\mathcal H_i$ is finite dimensional as well \cite{GierelarkFuture1}).% \Hm{??} 

Let $Q^A$ (respectively  $Q^B$) be the corresponding reduced density matrices
obtained from $Q$ by tracing out the corresponding degrees of freedom.
Then $Q^A \geq 0$, $\mathrm{Tr}_{\mathcal H_A } [Q^A]=1$, and identically in the case of $Q^B$. As is well
known the spectrum $\sigma  (Q^A) =( \lambda_n )$ is purely discrete (we are
presenting it always with the corresponding multiplicities and in
nonincreasing order) and in general different  from the spectrum of $Q^B$ in
the case of mixed states. For more on this see below and the Appendix~B.
In the case when, as in the introduction, $Q=|\Psi  \rangle\langle \Psi |$  for some $\Psi \in \mathcal H$  the
spectrum of $Q^A$ and $Q^B$ are equal to each other and equal to the list of
squared Schmidt coefficients of the corresponding Schmidt decomposition
of the vector $\Psi$ \cite{NC, BZ, 22}. The same is valid for the Hilbert--Schmidt level
reduced density matrices when we consider these type of Schmidt
decompositions of a given $Q \in \mathcal H$, see Appendix~B and \cite{PPAM2022, RGielerakFuture4}.

Let us recall now some well known facts on the reduced density matrices
Let $Q \in E(\mathcal H  )$. Let $\{ |i\rangle \}$ be an arbitrary complete orthonormal system of
vectors in $\mathcal H_B$. Then we have canonical unitary equivalence
\begin{equation}
\mathcal H_A \otimes \mathcal H_B \cong \oplus_i \mathcal H_A \otimes |i\rangle\,,
\end{equation}
where $\cong $ means that $\varphi \in \mathcal H$ is decomposed as $|\varphi \rangle \cong \oplus_i ( \mathrm{I}_{\mathcal H_A} \otimes | i\rangle \langle i |) |\varphi \rangle$.

%{ \color{red}\it -- TO-REMOVE --
%\subsection{ More on the reduced density matrices}
%
%Let  us  recall now  some   well known  facts  on the  reduced  density matrices. Let $Q \in E(\mathcal{H})$ .Let  $\{|i\rangle\}$ be an  arbitrary  complete  orthonormal  system of vectors  in  $\mathcal{H}_B$. Then  we have canonical, unitary   equivalence 
%\begin{equation}
%\mathcal{H}_A \otimes \mathcal{H}_B  \cong \oplus_{i=1} \mathcal{H}_A \otimes | i \rangle =  \oplus_{i=1} \mathcal{H}(i)
%\end{equation}
%where $\cong$ means  that   $| \varphi \rangle  \in \mathcal{H}$  is  decomposed  as
%\begin{equation}
%| \varphi \rangle \cong \sum_{i=1}^{????} |\varphi_i \rangle .
%\end{equation}
%-- TO-REMOVE --
%}

Then  for  any  $A \in B(\mathcal{H})$  we  can write  :
\begin{equation}
A = \left( \sum^{\infty}_{i=1} 1_{\mathcal{H}_A} \otimes | i \rangle\langle i | \right) (A) \left| \left( \sum^{\infty}_{i=1} 1_{\mathcal{H}_A} \otimes | i \rangle\langle i | \right) \right| = \sum^{\infty}_{i,j}A_{ij} ,
\end{equation}
where  
\begin{equation}
A_{ij} = ( 1_{\mathcal{H}_A} \otimes |i \rangle\langle i| ) (A)  ( 1_{\mathcal{H}_A} \otimes | i \rangle \langle i | ) ,
\end{equation}
is the bounded linear map from $\mathcal{H}_A \otimes | i \rangle$ to $\mathcal{H}_A \otimes | j \rangle$.

Using  the  Krauss  decomposition  Theorem~\ref{eq-Kraus} we  have  the  following  observation : The  linear  and  bounded  map 
\begin{equation}
\begin{array}{l}
	\mathrm{Tr}_B : L_1(\mathcal{H}) \mapsto L_1(\mathcal{H}_A) , \\
		A \mapsto \mathrm{Tr}_B (A)  \cong \sum^{\infty}_{i=1} A_{ii} , \\
\end{array}
\label{lbl:eq:trace:b:and:map:A}
\end{equation}
named   partial  trace  map  is  quantum operation in  the  sense  of   the  previously  introduced definition in Section~\ref{lbl:sec:majorisation:theory}.
\begin{theorem}
Let  $\mathcal{H}_A$ and  $\mathcal{H}_B$ be  a pair  of  separable  Hilbert  spaces  of  an  arbitrary  dimensions $\mathcal{H}= \mathcal{H}_A \otimes \mathcal{H}_B$ and  let  $Q \in E(\mathcal{H})$ and let  $Q^A= \mathrm{Tr}_B (Q) \in E ( \mathcal{H}_A )$ and  $Q^B= \mathrm{Tr}_A(Q) \in E( \mathcal{H}_B )$ be the corresponding  reduced  density matrices.Then:
\begin{equation}
\begin{array}{l}
FEN_{\pm} (Q^A ) \leq FEN_{\pm} (Q), \\
FEN_{\pm} (Q^B ) \leq FEN_{\pm} (Q).
\end{array}
\end{equation}	
\end{theorem}
\begin{proof}
Follows  from  the    formula~\ref{lbl:eq:trace:b:and:map:A} which  demonstrates  that  the  operations  of  taking  partial  traces are quantum  operations  and application of Theorem~\ref{lbl:thm:38}.
\end{proof}

Let $\mathcal{H}= \mathcal{H}_A\otimes \mathcal{H}_B$ be a bipartite separable Hilbert space and let $Q \in E(\mathcal H )$. It is
well known that the spectrum of $Q$ counted with multiplicities, denoted $\sigma (Q)
=(\lambda_1,...)$ is purely discrete and the following spectral decomposition holds:
\begin{equation}\label{eq-spectralresol}
 Q =\sum_{n=1}^\infty \lambda_i| \Psi_n  \rangle\langle \Psi_n|\,,
\end{equation}
where the orthogonal (and normalised) system of eigenfunctions $|\Psi_n \rangle$ of $Q$
forms a complete system.
%It is assumed below that the spectrum of is simple, i.e. all eigenvalues of $Q$ are
%non-degenerated. The general case will be presented in another paper [ ]....\Hm{back}
Each eigenfunction $|\Psi_n \rangle$ can be expanded further by the use of the Schmidt
decomposition:
\begin{equation}\label{eq-subs1}
| \Psi_n \rangle  = \sum_{i=1}^\infty \tau_i^n | \psi^n_i \otimes \phi_i^n  \rangle\,,
\end{equation}
where $\tau_i^n \geq 0$ $\sum_{n=1}^\infty (\tau_i^n)^2 =1 $ and the systems $ \{\psi^n_i \}$
and $ \{\phi^n_i \}$ form the complete orthonormal systems in $\mathcal H _A$ and, respectively, $\mathcal H _B$. Using~(\ref{eq-subs1}) and~(\ref{eq-spectralresol}) we can compute the corresponding reduced density matrices
\begin{equation}\label{eq-subs}
Q^B= \mathrm{Tr}_A  \left[\sum_{i=1}^\infty \lambda_n |\Psi_n  \rangle\langle \Psi_n| \right] =  \sum_{n=1}^\infty \lambda_n Q_n^B\,,
\end{equation}
where the operators 
\begin{equation}\label{eq-subs2}
Q_n^B=\sum_{i=1}^\infty  |\tau_n|^2 |\phi_i ^n  \rangle\langle \phi_i ^n  |
\end{equation}
are the states on $\mathcal H_B$.
Similarly, for the reduced density matrix connected to the observer localized
with $\mathcal H_A$:
\begin{equation}
Q^A  = \mathrm{Tr}_B (Q) = \mathrm{Tr}_B \left[\sum_{n}^\infty \lambda_n | \Psi_n   \rangle \langle \Psi_n | \right] =  \sum_{n=1}^\infty \lambda_n Q_n^A ,
\end{equation}
where
$Q_n^A=\sum_{i=1}^\infty |\tau_i^n|^2 | \psi^n_i  \rangle\langle \psi^n_i|$ are states on $\mathcal H_A$.
The obtained systems of operators $\{ Q_n^A\}$ and $\{ Q_n^B\}$ consist of bounded non-negative self-adjoint, local operators of class $L_1(\mathcal H_A)$, respectively  of class $L_1(\mathcal H_B )$
and therefore they  are locally measurable. In particular the squares of
the Schmidt coefficients $\tau_i^n$  of the Schmidt decompositions of the eigenfunctions of the parent state $Q$ are observable (measurable locally) quantities.

\begin{proposition}
Let  $\mathcal{H}= \mathcal{H}_A \otimes \mathcal{H}_B$  be  a  bipartite  separable  Hilbert  space and
let  $Q \in E(\mathcal{H})$.
Let   $(Q^A, Q^B)$ be the  corresponding  reduced  density matrices and  let
\begin{equation}
	Q_n^A=\sum_{i=1}^{\infty}|\tau_i^n|^2 |\psi_i^n \rangle \langle \psi_i^n| .
\end{equation}
And  corr.
\begin{equation}
	Q_n^B=\sum_{i=1}^{\infty}|\tau_i^n|^2 |\phi_i^n  \rangle\langle \phi_i^n| ,
\end{equation}
be  the  corresponding operators as  defined  in  (2.0 , corr. In 2.).
Then,  for  any  $n$:
\begin{enumerate}
\item $G(Q^A(n)) = \det (\mathrm I_{\mathcal H _A }+ Q^A(n)) = \prod_{j=1}^{\infty}(1+(\tau_j^n)^2) \leq e$,
\item The  value    $G(Q^A(n))$ is  invariant under  the  action  of  unitary  group,
  for  any unitary  map  $U  \in \mathcal{H}_A$:
\begin{equation}
	G(UQ^A(n)U^{\dagger}) = G(Q^A (n))
\end{equation}            
\item The value $G(Q^A  (n))$ is  not  increasing  under  the  action  of  any  local quantum  operation
$\Phi$ acting on  $E( \mathcal{H}_A )$:
\begin{equation}
G(\Phi (Q^A(n) )) \leq G(Q^A(n))
\end{equation}  
\end{enumerate}
Identical   facts  are  valid  for  the  reduced  density matrices  $Q^B(n)$ .
\end{proposition}

\begin{proof}
Obvious.
\end{proof}

\begin{remark}
\rm{The  list $\Gamma(Q) = (r_j^n)$  associated  with  $Q$  is   locally  $SU(\mathcal{H}_A)\otimes SU(\mathcal{H}_B)$
matrix valued  invariant of $Q$  (after   taking  care  on the localisation in  this  $2d$  table of the corresponding
Schmidts  numbers). Therefore  any  scalar  functions  build on  $\Gamma$ will define  a  locally-unitary invariant of  $Q$.
Some  of  them  are  additionally also monotonous  under  the  action  of  the  local  quantum  operations and  therefore  are
promising  candidates  for being  a  "good" \cite{NC, BZ, Guhne2008, HorodeckiReview} quantitative  measures of  quantum correlations included  in  $Q$.
More on this is reported elsewhere \cite{Gielerak2020, Gielerak2021}.

Another  approach  to   certain version  of   reduced  density matrices  structure  is  based  on  the  use  of  the  Schmidt
decomposition  method in the  Hilbert-Schmidt  space  of  operators build  on  the  space
$\mathcal{H}_A\otimes \mathcal{H}_B$. Some details  are presented  in  appendix  B and in paper \cite{PPAM2022}.

Systematic and  much  wider   applications  of  the  obtained  forms  of  the  reduced  density  matrices  will  be  presented
in an  another  publications (under  preparations  now).}
\end{remark}

\subsection{The  case  of  pure  states}
Let  $\mathcal{H}= \mathcal{H}_A \otimes \mathcal{H}_B$  be  a  bipartite, separable  Hilbert  space and
let  $Q \in E(\mathcal{H})$  be  such that  tr$(Q^2)=1$. Then  there  exists  an unique, normalized  vector
$|\Psi \rangle \in \mathcal{H}$  such that  $Q=|\Psi \rangle\langle \Psi|$.

Let  $\{e_i^A, i=1, \dots\}$, resp. $\{e_j^B, j=1, \dots\}$ be  some  complete  orthonormal systems  in
$\mathcal{H}_A$, resp. in $\mathcal{H}_B$.

Then  we  can  write :

\begin{equation}
|\Psi \rangle=\sum_{i,j=1}^{\infty} \Psi_{ij}|e_i^A \rangle\otimes |e_j^B\rangle
\end{equation}

where  $\Psi_{ij}=\langle e_j^B \otimes e_i^A|\Psi\rangle$.

We  start  with  the  Schmidt  decomposition  (essentially  SVD decomposition, see i.e. Thm.~26.8 in \cite{BB})  in  the  infinite dimensional setting.

\begin{theorem}
For any  unit  vector   $|\Psi\rangle \in  \mathcal{H}$  there  exist
\begin{itemize}
\item sequence of non-negative numbers $\tau_n$ (called  the  Schmidt  coefficients  of  $\Psi$) and such that $\sum_{n=1}^{\infty} \tau^2_n = 1$,

\item two ,complete orthonormal  systems  of  vectors $\{\phi_n\}$ in $\mathcal{H}_A$ and
$\{\omega_n\}$   in   $\mathcal{H}_B$  such that  the  following  equality (in  the  $L_2$-space  sense) holds:
\begin{equation}
|\Psi\rangle=\sum_{n=1}^{\infty} \tau_n|\phi_n \rangle|\omega_n \rangle \,.
\label{lbl:eq:psi:decomposition}
\end{equation}
\end{itemize}
\end{theorem}

The  decomposition  \ref{lbl:eq:psi:decomposition} is  called  the  Schmidt  decomposition  of $| \Psi \rangle$.
The  expansion  formula  \ref{lbl:eq:psi:decomposition} can  be  rewritten as:
\begin{equation}
|\Psi\rangle=\sum_{i=1}^{\infty}|e_i^A \rangle |F_i^B \rangle
\end{equation}
where
\begin{equation}
|F_i^B \rangle=\sum_{j=1}^{\infty}\Psi_{ij} |e_j^B\rangle
\label{Fbi:infinite:decomposition}
\end{equation}
and  also
\begin{equation}
|\Psi \rangle =\sum_{j=1}^{\infty}|F_j^A\rangle\,|e_j^B\rangle\,,
\label{psi:infinite:decomposition}
\end{equation}
where
\begin{equation}
|F_j^A \rangle=\sum_{i=1}^{\infty}\Psi_{ij} |e_j^A \rangle\,.
\label{Faj:infinite:decomposition}
\end{equation}
Let us  define  pair of linear   maps  $J^A: \mathcal{H}_A \rightarrow  \mathcal{H}_B$,
resp. $J^B: \mathcal{H}_B \rightarrow  \mathcal{H}_A$ by the  following
\begin{equation}
J^A\,:\,|e_i^A  \rangle\rightarrow |F_i^B\rangle
\end{equation}
and  then extended  by  linearity  and  continuity  to  the  whole  $\mathcal{H}_A$.
In an  identical  way  the  map  $J^B$ is  defined. Both  of the introduced   operators  $J$  are bounded as  can  be
seen  by  simple  arguments. Now  we  define a  pair  of  operators  which  plays  an important  role  in  the  following
\begin{equation}
\Delta^A(\Psi) \, : \, J^{A^{\dagger}}J^A: \mathcal{H}_A \rightarrow \mathcal{H}_A
\end{equation}
and  similarly
\begin{equation}
\Delta^B(\Psi) \, : \, J^{B^{\dagger}}J^B: \mathcal{H}_B \rightarrow \mathcal{H}_B
\end{equation}

Some elementary properties of the introduced operators $\Delta^A$ and $\Delta^B$ are collected in the following proposition.

\begin{proposition}
The  operators $\Delta^A$ and  $\Delta^B$   have  the  following  properties:

\begin{itemize}
\item[RDM1] They  both  are  non-negative  and  bounded $\|\Delta^A \|_1 = \|\Delta^B \|_1 =1$.
\item[RDM2] The  non-zero parts of  the  spectra of  $\Delta^A$ and  $\Delta^B$ coincides and are  equal  to  squares $\tau^2_n$ of  non-zero Schmidt numbers in (\ref{lbl:eq:psi:decomposition}).
%Let $S_{>}(\Psi)$  be  a list  of  non-zero Schmidt  coefficients  of  the  vector $|\Psi \rangle $, see (Th. 3.1).  Then  the list  $S_{>}(\Psi)$  {\color{red}coincides}  with  non-zero elements of  lists  of squared nonzero  eigenvalues  of  the  operators $\Delta^A(\Psi)$ and $\Delta^B(\Psi)$ which  are  identical. \Hm{back Th}
\item[RDM3] In particular the following formulas are valid:
\begin{equation}
	\begin{array}{lcl} \nonumber \Delta^A |\phi_n \rangle &=&  \tau_n^2 |\phi_n \rangle , \\ \\
		\Delta^B |\omega_n \rangle &=&  \tau_n^2 |\omega_n \rangle\,,
	\end{array}
\end{equation}
which  means  that  the  kets  $|\phi_n \rangle$ are  eigenvectors  of  the  reduced  density  matrix $Q^A$,
and  similarly  for $Q^B$.
\end{itemize}
\end{proposition}

The  interesting  observation  is  that  the  explicite  Gram  matrix  nature (it  is  well known  fact \cite{20} that any (semi)-positive  matrix  has  a  Gram  matrix  structure) of  the  operators  $\Delta$ can be  flashed  on.

\begin{proposition}
Let
\begin{equation}
|\Psi\rangle =\sum_{i,j=1}^{\infty} \Psi_{ij}|e_i^A \rangle \otimes |e_j^B \rangle  \in \mathcal{H} ,
\end{equation}
be  given. Then  the  matrix  elements of  the  corresponding operators $\Delta$, given in  the  product  base
$\e_i^A \rangle \otimes |e_j^B \rangle $ are  given  by  the  formulas below
\begin{equation}
\Delta_{ij}^A(\Psi)=\langle e_j^A|\Delta^Ae_i^A \rangle _{\mathcal H_A}=\langle F_j^B|F_i^B\rangle _{\mathcal H_B}
\end{equation}
and  similarly
\begin{equation}
\Delta_{ij}^B(\Psi)=\langle e_j^B|\Delta^Be_i^B \rangle _{\mathcal H_B}=\langle F_j^A|F_i^A\rangle _{\mathcal H_A}
\end{equation}
where  the corresponding  vectors  $F$  are  given by (\ref{Fbi:infinite:decomposition})  and  (\ref{Faj:infinite:decomposition}).
\end{proposition}

In  the  finite  dimensional  case  the  following,  niece geometrical  picture  is  known \cite{Gielerak2020}.
Let  $\{v_i, i=1,..d’\}$ be  a  system  of  linearly independent  vectors in  the  space $C^d$, where $d’=d$.
Let  as  build  on  the  vectors  $d$ dimensional  parallelepiped. Then, the  Euclidean  volume  of  this parallelepiped is
equal  to the  determinant  of  the  Gram  matrix  built on  these vectors. The  matrix  elements  of  this Gram  matrix  are
given  by  the  scalar  products $\langle v_i|v_j \rangle $ for  $i,j=1:d’$. Under  the  condition  that  the  sum of  the  lengths of the
spanning vectors $v_i$  is  equal to  1 the  parallelepiped which has  the maximal  volume  is  that  which  is  spanned  by
the  system  of  orthogonal  vectors  of  equal  length. In  this  particular  case  the  corresponding  Gram matrix elements
are  equal  to $(1/d’)\delta_{ij}$. In a  general  case  the  volume  of  the  parallelepiped spanned  by  the  vectors
forming  some square  matrix columns (or  rows)  can be  estimated  from  above  be  several  inequalities . The  Hadamard
inequality saying  that  this  volume  is  no bigger  then  the  product  of  the  lengths  of the  spanning vectors $v_i$  is the  best known  among them. For  more  on  this  see \cite{Gielerak2020}.

On  the  basis  of  results and  facts  presented in previous  sections  we  can  define   the  following  quantity (in fact
entire  function of $z$)  that  will be  called gramian  function  of  the  state $|\Psi \rangle $.

\begin{equation}
    G(\Psi)(z)=  \det (\mathrm I_A +z\Delta^A (\Psi))=
    \det ( \mathrm I_B +z\Delta^B (\Psi)) =\prod_{n=1}^{\infty}(1+z\tau_n^2) .
    \label{lbl:eq:G}
\end{equation}
In  particular  case $z=1$  the  value  of  the  gramian  function $G$  of  state $|\Psi\rangle $
will be  called  the  gramian  volume  of $|\Psi\rangle $  and  denoted  as  $G(\Psi)$. The  logarithm  of  the  gramian
volume  will be  called  the  logarithmic (gramian) volume  of $|\Psi\rangle $ and  denoted  as  $g(  \Psi )$. Using  (\ref{lbl:eq:G})  it  follows  that
\begin{equation}
    g (\Psi)=\sum_{n=1}^{\infty} \log(1+\tau_n^2) .
\end{equation}

\begin{proposition}
Let
\begin{equation}
|\Psi \rangle  = \sum_{i,j=1}^{\infty} \Psi_{ij}|e_i^A \rangle \otimes |e_j^B\rangle  \in \mathcal{H}.
\end{equation}
Then the gramian volume $G(\Psi)$ has the following properties:
\begin{enumerate}
\item For  any $|\Psi \rangle : ~~ 2 \leq G(\Psi) \leq e$.
\item $G(\Psi) =2$  iff  $\Psi$  is  a  separable  state, i.e.  Schmidt  rank  of  $\Psi$ is  equal  to  1.
\item Let $\mathcal{U}(\mathcal{H})$  be  a multiplicative  group of  unitary  operators  acting  in the
   Hilbert  space  $\mathcal{H}$. Then  the  gramian  volume  of $|\Psi \rangle $   is  invariant  under the  action on
   of  the local unitary groups $\mathcal{U}(\mathcal{H}_A)\otimes \mathcal{U}(\mathcal{H}_B)$.
\item Let  $\Phi_{A(B)}$ be  any local  quantum  operation  on the  local space  $\mathcal{H}_A$
    (resp. $\mathcal{H}_B$). Then
\begin{enumerate}
    \item $G( (\Phi_A \otimes \mathrm I_B) (\Psi)) \leq G(\Psi)$,
    \item $G( (\mathrm I_A \otimes \Phi_B) (\Psi)) \leq G(\Psi)$,
    \item $G( (\Phi_A \otimes \Phi_B) (\Psi)) \leq G(\Psi)$.
\end{enumerate}
\end{enumerate}
\end{proposition}

We are seeing that the gramian volume of pure  states  is locally  invariant (under  the  local  unitary  operations
action)  quantity.  And  what  is   also important we  have  proved  that the  gramian  volume defined  in (\ref{lbl:eq:G}) do not
increase  under  the  action  of  any  separable  quantum  operation. This is   why  we  think   that  the  gramian  volume
might  be  a  very  good  candidate  for  the   entanglement measure  included  in  pure quantum states.

We have  also  some  results  about  extensions  of  the  results  presented here  to the  many-partite  systems  as well \cite{GierelarkFuture3}.

Sometimes  it  is more useful  to use  logarithmic  gramian  volume  $g$ instead  of  the  gramian  volume  $G$. Some   basic
properties  of the  logarithmic  volume $g$ are  contained  in  the  following  Theorem.

\begin{theorem}
Let
\begin{equation}
| \Psi \rangle  = \sum_{i,j=1}^{\infty} \Psi_{ij}|e_i^A \rangle \otimes |e_j^B \rangle  \in \mathcal{H}.
\end{equation}
Then  the logarithmic   gramian  volume $g(\Psi)$   has  the  following  properties;
\begin{enumerate}
\item For  any   $|\Psi \rangle $
\begin{equation}
     g (\Psi)= \sum_{n=1}^{\infty} \log(1+\tau_n^2),
\end{equation}
where  $\tau_n$ are  the  Schmidt  numbers  of  $|\Psi \rangle $.
\item For  any $| \Psi \rangle $:
\begin{equation}
 \log(2) \leq g( \Psi) \leq 1 .
\end{equation}
\item  $g(\Psi) =\log(2)$ iff  $\Psi$  is  a  separable  state, i.e.  Schmidt  rank  of  $\Psi$ is  equal  to  1.
\item  Let  $\mathcal{U}(\mathcal{H})$ be  a multiplicative  group of  unitary  operators  acting  in the  Hilbert
space  $\mathcal{H}$. Then  the logarithmic  gramian  volume  of $|\Psi \rangle $  is  invariant  under the  action on  of  the local
unitary groups
\begin{equation}
 \mathcal{U}(\mathcal{H}_A)\otimes \mathrm I_B, ~~ \mathrm I_A \otimes \mathcal{U}(\mathcal{H}_B),
~~ \mathcal{U}(\mathcal{H}_A) \otimes \mathcal{U}(\mathcal{H}_B)
\end{equation}
\item Let  $\Phi_{A(B)}$ be  any local  quantum  operation  on the  local space
$\mathcal{H}_A$ (resp. $\mathcal{H}_B$). Then
\begin{enumerate}
    \item $g( (\Phi_A \otimes \mathrm I_B) (\Psi)) \leq g(\Psi)$,
    \item $g( (\mathrm I_A \otimes \Phi_B) (\Psi)) \leq g(\Psi)$,
    \item $g( (\Phi_A \otimes \Phi_B) (\Psi)) \leq g(\Psi)$.
\end{enumerate}
\end{enumerate}
\end{theorem}

Similar results  are  true  for  the  renormalized  von Neumann entropies. For  this  goal  let  us  recall  the  definitions
of  entropies  generating  operators :
\begin{equation}
S_{-}(\Psi) = (\mathrm I_{\mathcal{H}_A} + \Delta^A(\Psi))^{-(\mathrm I_{\mathcal{H}_A} + \Delta^A(\Psi))} - \mathrm I_{\mathcal{H}_A} = \sum_{n=1}^{\infty} \bigg( (1+\tau_n^2)^{-(1+\tau_n^2)} - \mathrm{1} \bigg) | \phi_n \rangle \langle \phi_n | .
\label{lbl:eq:s:minus}
\end{equation}
From  which  we  obtain  an  estimate

\begin{lemma}
For any pure state $|\Psi \rangle  \in \mathcal{H}$ the renormalized entropy generator defined as $S_{-}(\Psi)$ in (\ref{lbl:eq:s:minus}) obeys the bound:
\begin{equation}
\|S_{-}(\Psi)\|_1 \leq 2\|\Delta^A\|_1 = 2 .
\label{lbl:eq:s:minus:psi}
\end{equation}
\end{lemma}

\begin{proof}
We  have  used  the  following, rough  estimate:
\begin{equation}
|1-(1+\tau_n^2)^{-(1+\tau_n^2)}| \leq  (1+\tau_n^2) \log(1+\tau_n^2) \leq (1+\tau_n^2)\tau_n^2
\end{equation}
From  which  the  bound  (\ref{lbl:eq:s:minus:psi})  follows immediately.
\end{proof}

\begin{theorem}
Let $\mathcal{H} =\mathcal{H}_A \otimes  \mathcal{H}_B$, where  $\mathcal{H}_A$ and  $\mathcal{H}_B$
are  separable  Hilbert spaces. Then, for  any pure state $|\Psi \rangle  \in E(\mathcal{H})$ the  renormalized  entropy
defined  as
\begin{equation}
\mathrm{FEN} (\Psi) = \log (  \det ( \mathrm I_A + S_{-} (\Psi )  ) = \sum_{n=1}^{\infty} (1+\tau_n^2) \log(1+\tau_n^2) ,
\end{equation}
is  finite, $L_1$-continous on $E(Q)$  and  bounded by :
\begin{equation}
           0 \leq  \mathrm{FEN} (  \Psi  ) \leq 2.
\end{equation}           
\end{theorem}

\begin{theorem}
Let  $\mathcal{U}(\mathcal{H})$ be  a multiplicative  group of  unitary  operators  acting  in the  Hilbert
space  $\mathcal{H}$. Then  the renormalized entropy of $|\Psi \rangle $  is  invariant  under the  action on  of  the local
unitary groups  $ \mathcal{U}(\mathcal{H}_A)\otimes \mathrm I_B, ~~ \mathrm I_A \otimes \mathcal{U}(\mathcal{H}_B)$ and also
$\mathcal{U}(\mathcal{H}_A) \otimes \mathcal{U}(\mathcal{H}_B)$.
\end{theorem}

\begin{theorem}
Let $\mathcal{H} =\mathcal{H}_A\otimes \mathcal{H}_B$, where  $\mathcal{H}_A$ and
$\mathcal{H}_B$  are  separable  Hilbert spaces.
Then, for  any pure state $|\Psi \rangle  \in E(\mathcal{H})$ the  renormalized  entropy  defined  as
\begin{equation}
	(\Psi) = \log ( \det ( \mathrm I_A + S_{-} ( \Psi )  ),
\end{equation}
is nonincreasing under  the  action of  any local  quantum  operation $\Psi_{A(B)}$  on the  local space
$\mathcal{H}_A$ (resp. $\mathcal{H}_B$).
\end{theorem}

Finally  we  mention  the  monotonicity of  the  renormalised Entropy with  respect  to  the by majorisation relation introduced  semi-order. Details  will be presented  elsewhere \cite{RGielerakFuture4}.

\appendix

%\renewcommand{\theequation}{\thesection.\arabic{equation}}

%\section{Appendix A. Fermionic  Fock  space  aspects}
\section{Fermionic  Fock  space  aspects} \label{lbl:app:sec:fock:spaces}

\renewcommand{\theequation}{\thesection.\arabic{equation}}
\setcounter{equation}{0}

Let $\mathcal{H}$  be  a  separable  Hilbert  space  and  let   the  system of vectors $\{e_n \}_{n=1,\dots }$
forms  a  complete  orthonormal  system (i.e.  the  orthonormal base) in $\mathcal{H}$. Then  the  system
$\{e_{i_1}\otimes \cdots \otimes e_{i_n} \}$ forms an  orthonormal base  in $\mathcal{H}^{\otimes n}$.
The  free  Fock  space  over  $\mathcal{H}$  is  defined as:
\begin{equation}
\Gamma(\mathcal{H})=\oplus _{n=0}^{\infty} \mathcal{H}^{\otimes n}
%\eqno{(A.1)}
\label{lbl:eq:Gamma}
\end{equation}
where  $\mathcal{H}^{\otimes  0}= \mathcal{C}$.

The anti-symmetrization operator $\wedge^n$ on the  $n$-fold  summand of Eq.~\ref{lbl:eq:Gamma}:
\begin{equation}
\wedge^n(f_1\otimes \cdots \otimes  f_n) = \frac{1}{n!} \sum_{\pi \in S_n} (-1)^s{(\pi)}
f_{\pi(1)} \otimes  \cdots \otimes f_{\pi(n)}
%\eqno{(A.2)}
\equiv f_1  \wedge \dots \wedge f_n ,
\end{equation}
where $S_n$ stands  for  symmetric  group of  order  $n$  and  $s(\pi)$ stands  for  the  parity  of $\pi$.

Operator  $\wedge^n$ is then  extended  by linearity  and  continuity and  normalised  properly  to be  the  orthogonal
projector  acting in  the free Fock  space  and  with  the  range  which  is  called  the  fermionic  Fock  space  over
$\mathcal{H}$  and  denoted  as
\begin{equation}
\wedge(\Gamma(\mathcal{H})) = \oplus_{n=0}^{\infty} \wedge^n (\mathcal{H}^{\otimes n}).
%\eqno{(A.3)}
\end{equation}
In  particular   the  system  $\{e_{i_1}\wedge \dots \wedge e_{i_n} \}$ forms an orthonormal
base in $\wedge(\mathcal{H}^{\otimes n})$.

\begin{lemma}
For  any  tensors  $F=f_1 \otimes \cdots \otimes f_n$ and
$G = g_1 \otimes  \cdots \otimes g_n$ we have
\begin{equation}
\begin{array}{lcl}
\langle f_1 \wedge \dots \wedge f_n |g_1 \wedge \dots \wedge g_n \rangle %\eqno{(A.3)}
& = & \frac{1}{(n!)^2} \sum_{\pi, \pi' \in S_n}(-1)^{s(\pi) + s(\pi')} \prod_{i=1}^n
 \langle f_{\pi(i)}|g_{\pi'(i)}\rangle \\ \\
& = & \frac{1}{n!}\det(\mathcal{R}(FF|GG)), 
\end{array}
\end{equation}
where  $\mathcal{R}(FF|GG)$ is the relative Gramian matrix build on $FF=\{f_1, \dots, f_n\}$
and $GG=\{g_1, \dots, g_n\}$, see \cite{Gielerak2020} for the corresponding definitions.
\end{lemma}

The  fermionic  Fock  space  over  $\mathcal{H}$ is  defined  as
\begin{equation}
\Gamma_{as}(\mathcal{H}) \cong \oplus _{n=1}^{\infty} \wedge^n (\mathcal{H}),
%\eqno{(A.4)}
\end{equation}
i.e.  $\Psi=\oplus _{n=1}^{\infty} | \psi_n \rangle \in \Gamma_{as}(\mathcal{H})$,
$| \psi_n \rangle \in \wedge^n(\mathcal{H})$ iff $\sum_{n=0}^{\infty} n! \|\psi_n \|^2 < \infty$.

\begin{remark}
If  $\dim (\mathcal{H}) = d < \infty$  then $\wedge^n(\mathcal{H})=\emptyset$  for $n>d$  and
$\dim(\wedge^n(\mathcal{H}))={d \choose n}$ for  $n \leq d$.
The  corresponding antisymmetric  Fock  spaces  in this  situation are used for  describing  fermionic,  discrete  degrees
of  freedom .
\end{remark}

Let  $T \in B(\mathcal{H})$. Then  we  lift  the  action of $T$  onto  the  Fock  space(s) as
\begin{equation}
  \Gamma(T): f_1\otimes \cdots \otimes f_n \rightarrow Tf_1\otimes \cdots \otimes Tf_n ,
%\eqno{(A .5)}
\end{equation}
and  similarly  for $f_1\wedge \dots \wedge f_n$  case.

Let  us  collect  here  some  well known  facts:

\begin{proposition}
Let $T \in B(\mathcal{H})$, then
\begin{enumerate}
\item $\Gamma(T) \in B (\Gamma_{*}(\mathcal{H}))$  (where * stands  for  empty sign  or  as)
\item for $T ,S \in B (\mathcal{H}))$,
\item $\Gamma (TS) = \Gamma(T ) \Gamma(S)$.
\end{enumerate}
\end{proposition}

\begin{proposition}
Let  $T   \in L_1(\mathcal{H})$, then
\begin{enumerate}
\item  $\sigma(T^{\otimes n}) = ( \lambda_{i_1} \cdots \lambda_{i_n}, ~ \lambda_i \in \sigma(T) )$
\item $\sigma(T^{\wedge^n}) = \bigg( \lambda_{i_1} \cdots \lambda_{i_n}, ~ i_1 < \cdots < i_n, ~ \lambda_i \in \sigma(T) \bigg)$
\item  $\mathrm{Tr}[T^{\otimes n}]=(\mathrm{Tr}[ T])^n$
\item $\mathrm{Tr}[ T^{\wedge n}] = \sum_{i_1<\cdots<i_n}
    \lambda_{i_1}\cdots \lambda_{i_n} = \frac{1}{n!} (\mathrm{Tr}[T])^n ~~ $ (for $T\geq 0$)
\end{enumerate}
\end{proposition}

\begin{corollary}
Let  $T \in L_1 (\mathcal{H})$, then
\begin{enumerate}
\item  $\|T^{\otimes n}\|_1 \leq  \| T \|_1^n$
\item  $ \|T^{\wedge n} \|_1 \leq \frac{1} {n!} \| T \|_1^n$
\item  $\Gamma (T) \in L_1(\Gamma(\mathcal{H}))$   if $\| T\|_1 < 1$
\item  $\Gamma_{as} (T)\in L_1(\Gamma_{as}(\mathcal{H}))$    iff  $\| T \|_1  < \infty$
\end{enumerate}
\end{corollary}

Now  we are  in the  position  to  prove  the  Grothendick  result  about  the  possibility  to determine  the Fredholm
determinant  of  infinite-dimensional matrices by  the  second-quantisation  mathematics methods.

\begin{theorem}[Grothendick] \label{A5GrothendickThm}
Let $T \in L_1 (\mathcal{H})$, then
\begin{equation}
	\det (\mathrm{I}_{\mathcal H} + T) = \sum_{n=0}^{\infty}\mathrm{Tr}[\wedge^n(T)]\, .	
\label{th-Agrod}
\end{equation}
\end{theorem}

From the Grothendick Theorem~\ref{A5GrothendickThm}  it  is  possible  to  prove in a relatively  easy  way \cite{S05} the  quoted  in  section~\ref{lbl:sec:Renorm:Neumann:Entropy} results on  Fredholm  determinants and  to  conclude  several  other  implications  of  his  ingenious  approach \cite{G56}. The key  point  is  the  following  observation.

\begin{lemma} %[A.6]
Let  $Q \in E (\mathcal{H})$ and  $\sigma(Q) =(\lambda_1,\dots )$ be  the  spectrum of $Q$.
Then
\begin{enumerate}
\item $\mathrm{Tr}[Q^{\otimes n}]=\sum_{i_1, \dots, i_n} \lambda_{i_1}\cdots \lambda_{i_n}=1$,
\item $\mathrm{Tr}[\wedge^n(Q)] \leq  \frac{1} {n!} \mathrm{Tr}[Q^{\otimes n} ] = \frac{1} {n!}$.
\end{enumerate}
\label{A6Lemma}
\end{lemma}

\begin{proposition}
Let $\mathcal{H}_{\oplus} = \mathcal{H}_A\oplus  \mathcal{H}_B$  be  a  two-particles separable  Hilbert  space.
Then, with  the  convention that “=”  means  temporary  the  unitary  equivalence  of  the  corresponding  spaces  the
following  is  true:
\begin{enumerate}
\item $\Gamma(\mathcal{H})=\Gamma(\mathcal{H}_A) \otimes \Gamma(\mathcal{H}_B)$,
\item $\Gamma_{as}(\mathcal{H}) = \Gamma_{as}(\mathcal{H}_A) \wedge \Gamma_{as}(\mathcal{H}_B)$,
\item $\Gamma_{sym}(\mathcal{H}) = \Gamma_{sym}(\mathcal{H}_A) \otimes _{sym} \Gamma_{sym}(\mathcal{H}_B)$.
\end{enumerate}
\label{A8Proposition}
\end{proposition}

\begin{proof} (sketch  of):
We  concentrate  only  on  the fermionic  case (2). The  reason is  that  it  is  the  case  which  is  relevant  for  the
purposes  of  the  present  note. The  point (1)  is  true  from the  general  fact  that  all separable, infinite
dimensional  spaces  are  unitary  isomorphic to each  other. The  case of  bosonic  space is  similar  to that  of
fermionic  space.

We  construct  a  special  unitary  map  $J$ from  the space
$\Gamma_{as}(\mathcal{H}_A) \wedge \Gamma_{as}(\mathcal{H}_B)$
to the  space  $\Gamma_{as}(\mathcal{H})$. For  this  goal  let us  observe  that:
%$$\Gamma_{as}(\mathcal{H}_A) \wedge \Gamma_{as}(\mathcal{H}_B)$$
%$$ = (\oplus_{N=0}^{\infty}(\wedge^N(\mathcal{H}_A)))\wedge(\oplus_{N=0}^{\infty}(\wedge^N(\mathcal{H}_B)))$$
%$$\sum_{N=0}^{\infty} (\oplus _{n+m=N} (\wedge^n(\mathcal{H}_A))\wedge (\wedge^m(\mathcal{H}_B))$$
\begin{equation}
\begin{array}{lcl}
\Gamma_{as}(\mathcal{H}_A) \wedge \Gamma_{as}(\mathcal{H}_B) & = & \bigg( \oplus_{N=0}^{\infty}\wedge^N(\mathcal{H}_A) \bigg) \wedge \bigg(\oplus_{N=0}^{\infty}\wedge^N(\mathcal{H}_B)\bigg) \\ \\
& & \oplus_{N=0}^{\infty} \Bigg( \oplus _{n+m=N} \; \bigg( \wedge^n(\mathcal{H}_A ) \bigg) \wedge \bigg( \wedge^m(\mathcal{H}_B) \bigg) \Bigg) .
\end{array}
\end{equation}
So, a typical $N$-particles  vector $\Psi_N$ looks in this  space  like: 
$\Psi_N = f_1\wedge \dots \wedge f_n \wedge g_1 \wedge \dots \wedge g_m$, where
$n+m=N$ and $f_i \in \mathcal{H}_A$ and $g_j \in \mathcal{H}_B$.
For  such  a  vector  we  define
\begin{equation}
J(\Psi_N ) = \left(\matrix{f_1 \cr 0}\right) \wedge \cdots \wedge \left(\matrix{f_n \cr 0}\right)
\wedge \left(\matrix{0 \cr g_1}\right) \wedge \cdots \wedge \left(\matrix{0 \cr g_m}\right)
\end{equation}
which  is  a  vector  from $\wedge^N(\mathcal{H}_A\oplus \mathcal{H}_B)$. It  easy to check  that  $J$
preserves  the  norm. Extending  $J$  by  linearity  and  continuity  argument  we  construct  the  unitary  map
\begin{equation}
	J:  \Gamma_{as}(\mathcal{H}_A)\wedge \Gamma_{as}(\mathcal{H}_B) \rightarrow \Gamma(\mathcal{H}_{\oplus}).
\end{equation}
\end{proof}

Let $T_A \in B (\mathcal{H}_A)$  and  respectively $T_B \in B (\mathcal{H}_B)$. Then $\Gamma(T_A) \in B( \Gamma_{as}(\mathcal{H}_A ))$ and  resp.
$\Gamma(T_B) \in B( \Gamma_{as} (\mathcal{H}_B ))$ and  the  norms. In  particular  if  $T_A$ and  $T_B$
are  of  trace  class  then  also $T_A\oplus T_B$  is  trace   class and  moreover:
\begin{equation}
\mathrm{Tr}[T_A \oplus T_B]  = \mathrm{Tr}[T_A ] + \mathrm{Tr}[ T_B ]\,.
\end{equation}
Note  that
\begin{equation}
T_A \oplus  T_B   = \left[ \matrix{T_A & 0 \cr 0 & T_B}\right],
\end{equation}
as  an  operator  acting  in $\mathcal{H}_A\oplus \mathcal{H}_B$.
\begin{theorem} \label{th-deter}
Let  $\mathcal{H}= \mathcal{H}_A \otimes \mathcal{H}_B$  be  a  bipartite  separable  Hilbert  space and  let
$T_A  \in L_1 (\mathcal{H}_A )$  and $T_B  \in L_1 (\mathcal{H}_B )$. The  following  formula  is  valid:
\begin{equation}
  \det ( 1  + T_A) \det(1+ T_B)  =  \det ( 1 + T_A\oplus T_B). 
\end{equation}
\end{theorem}
\begin{proof}
By  the  use  of  Grothendick Theorem~\ref{A5GrothendickThm} and Proposition~\ref{A8Proposition} pt.2, i.e.  the  use  of  the  unitary map  $J$  to
transport  the  operator $\Gamma(T_A \oplus T_B)$ onto the skew tensor product $\Gamma_{as}(\mathcal{H}_A)\wedge \Gamma_{as}(\mathcal{H}_B)$.
\end{proof}

\begin{corollary}\label{th-deter2}
Let $\mathcal{H}$ be  a separable  Hilbert space and  such that
$\mathcal{H} = \oplus_{i=1}^{\infty}\mathcal{H}_i$ and  let  $T \in L_1(\mathcal{H})$ be  of  the  form:
$T =  \oplus _{i=1}^{\infty}T_i$ (which implies  that  $\sum_{i=0}^{\infty}\|T_i\| _1<\infty$))
and $T_i(\mathcal{H}_i) \subseteq \mathcal{H}_i$, for all $i$. Then
\begin{equation}
	\det(\mathrm{I}_{\mathcal H}+T) = \prod_{i=1}^{\infty} \det(\mathrm{I}_{{\mathcal H}_i}+T_i)
\end{equation}
\end{corollary}

\begin{proof} Let  $P_N$  be  the  orthogonal projector in $\mathcal{H}$  onto  the  subspace
$\mathcal{H}_N=\oplus_{i=1}^N$. Applying  in the  inductive  way  (which  is  possible  due  to  associativity  of  the
procedures used  to  prove  Theorem~\ref{th-deter}  it  follows  that  the  following  is  true for any finite $N$:
\begin{equation}
\det( \mathrm{I}_{\mathcal H}+TP_N) = \prod_{i=1}^{N} \det(\mathrm{I}_{{\mathcal H}_i}+T_i) .
\label{lbl:eq:det:prod:for:N}
\end{equation}

The  existence  of  the $\lim_N$ (l.h.s of A) follows  from  the  $L_1$ continuity of  the  Fredholm  determinant formula~(\ref{lbl:eq:217}) from section~\ref{lbl:sec:Renorm:Neumann:Entropy}.
\end{proof}

\begin{example}
Let  $\mathcal{H} = \Gamma_{as}(h)$, where  $h$ is  some  separable Hilbert  space. Let  $T_n$  for  $n = 1,\ldots$ be  a  sequence  of  trace  class operators  defined  on  and  reduced by $\wedge^n(h)$.
Then  the  operator  $TT=\oplus _{n=1}^{\infty}T_n$ is  continuous  (on  the fermionic  Fock  space) and  of  the  trace
class  iff $\sum_{n=1}^{\infty} {|| T_n ||}_{1} < \infty$. From Theorem~\ref{th-deter} we learn that:
\end{example}
\begin{equation}
	\det(\mathrm{I}_{\Gamma_{as(h)}} + TT)= \prod_{i=1}^{\infty} \det(\mathrm{I}_{{\mathcal H}_i} + T_i).
\end{equation}
In  a  particular  case of  a given one-particle  operator  of  trace  class $T \in L_1(h)$ and  defining
$T_n = \wedge^n(T)$ we obtain
\begin{equation}
   \det (  1  +  \Gamma_{as} (T ) )  = \prod_{i=1}^{\infty} \det(\mathrm{I}_{\wedge^i(h)} + \wedge^i(T)).
\end{equation}
The following result might be also of some interest.

\begin{theorem}
Let  $\mathcal{H}$ be a separable  Hilbert  space and  let  $A \in L_1(\mathcal{H})$  and $B \in L_1 (\mathcal{H})$.
Then the   following  formula  is  valid:
\begin{equation}
	\det ( 1  + A)\det(1+B) = \det \bigg( (1+A)(1+B) \bigg) .
\end{equation}
\label{lbl:eq:T18}
\end{theorem}

\begin{proof}
If  both $A$  and  $B$  are  of  finite  range  the  proof  of (\ref{lbl:eq:T18}) follows  from the  corresponding  well known in  the
matrix  calculus result. Using  the  fact  that
finite range  operators  are  dense  in  $L_1(\mathcal{H})$  and  $L_1$-continuity  of
$\det (1+\cdots)$  the  proof  follows.
\end{proof}

%\section{Appendix  B. Schmidt  decompositions}
\section{Schmidt  decompositions}

\setcounter{equation}{0}

Let  $\mathcal{H}= \mathcal{H}_A \otimes \mathcal{H}_B$  be  a  bipartite  separable  Hilbert  space.
Then  the  space  $L_2(\mathcal{H})$ is  canonically  isomorphic  with  the  space
$L_2(\mathcal{H}_A) \otimes L_2(\mathcal{H}_B)$ as  is well known. In  particular , if  the  system of  operators
$\{ E_i^A \dots\}$, and resp.  $\{E_j^B \dots \}$ is   complete  and  orthonormal in
$L_2(\mathcal{H}_A)$, resp. in  $L_2(\mathcal{H}_B)$, then  the  system $\{ Ei^A  \otimes E_j^B\}$ is  complete
orthonormal system  in $L_2 (\mathcal{H})$ .

\begin{theorem} \label{ThmB1}
Let  $\mathcal{H}= \mathcal{H}_A \otimes \mathcal{H}_B$  be  a  bipartite  separable  Hilbert  space and let
$Q \in E (\mathcal{H})$. Then, there  exist -- a  system of  non-negative  numbers $(\tau_n)$,  $\sum_{n=1}^{\infty} \tau^2_n = {|| Q ||}^2_2$ called  the  canonical
($L_2$-space) Schmidt  numbers of  $Q$   and  such  that
-- two  complete, orthonormal systems of  $L_2$-class of  operators  $\{\Omega_n^A \} \subset L_2(\mathcal{H}_A)$, resp.
$\{\Omega_n^B \} \subset L_2(\mathcal{H}_B)$ such that:
\begin{equation}
Q = \sum_n \tau_n \Omega_n^A \otimes \Omega_n^B
\label{lbl:eq:B1}
\end{equation}
\end{theorem}

Let  $He (\mathcal{H})$  be  the  real Hilbert   space  of  $L_2$-class and  additionally  hermitean operators  acting
in the  space $\mathcal{H}$. In particular  $E(\mathcal{H})$  is  subset  of $He(\mathcal{H})$.
As  the  SVD  theorem  and  the  spectral  theorem  are  still valid  in the  space
$He(\mathcal{H})$ \cite{Moretti2016, Oreshina2017}  we  can  decompose  any state  $Q$  in  this  space in the spirit of  Schmidt decomposition.

\begin{theorem} \label{ThmB2}
Let  $\mathcal{H}= \mathcal{H}_A \otimes \mathcal{H}_B$  be  a  bipartite  separable  Hilbert  space and let $Q \in E (\mathcal{H})$. Then  there  exist -- a  system of  non-negative  numbers $(\tau^{\star}_n )$,  called  the hermitean, Schmidt  numbers of  $Q$   and  such  that $\sum_n (\tau_n^\star)^2 = \|Q\|_{He}^2$ and two complete (in the corresponding $He$ spaces), orthonormal systems of  $L_2$-class hermitean operators $\{ \Omega_n^{\star A} \} \subset He (L_2(\mathcal{H}))$, resp. $\{ \Omega_n^{\star B} \} \subset He(L_2(\mathcal{H}))$ such that:
\begin{equation}
Q = \sum_{n=1}^{\infty} \tau^\star_n \Omega_n^{\star A} \otimes \Omega_n^{\star B} .
\label{lbl:eq:B2}
%\eqno{(B.2)}
\end{equation}
\end{theorem}

\begin{remarks}
Whether  the  Schmidt  numbers  of  both  expansions  are identical or not is not clear for us. Also  the  operators
$\Omega$ appearing  in  {Theorem~\ref{ThmB1} and Theorem~\ref{ThmB2} are different  in general}. In  particular, all  the  operators  appearing in (\ref{lbl:eq:B2})  are  hermitean.
\end{remarks}

\begin{corollary}
If all the operators appearing in Eq.~(\ref{lbl:eq:B2}) are  non-negative  then  $Q$  is  separable.
\end{corollary}

\begin{proof}
If  dimensions  of  the  spaces  $\mathcal{H}_A$ and  $\mathcal{H}$  are both  finite then the  proof  follows  from  the
very  definition  of  separability.
For  $N < \infty$  we define  (modulo  normalisation)  using  expansion  (B.2)  the  following  separable  states:
\begin{equation}
Q^N = \sum_{n=1:N} \tau_n^{\star} \Omega_n^{\star A} \otimes \Omega_n^{\star B} .
\label{lbl:eqn:B3}
\end{equation}
The  sequence  $Q^N$  tends  in  the  $L_2$  topology  to  the  limiting  state  $Q$. Therefore  we  conclude that  $Q$
belongs  to  the  $L_2$ closure  of  the  set of  separable  states.  But  $Q$    belongs  to  $E(\mathcal{H})$ from  the
very  assumptions made  on it.
\end{proof}

As  it is  well  known  the  Schmidt  decompositions  (B.1)   and  (B.2)  can  be  used in  finite  dimensions to  test the   presence  of  entanglement  in  $Q$. For  this, let us recall the  well  known  realignment  criterion: if  $Q$  is  separable  then the sum  of the corresponding canonical Schmidt numbers $\tau$ is not bigger then 1 \cite{Rudolph2003, Rudolph2005}.

For  other, generalized  version  of  this  criterion see \cite{Zhang2017, Zhang2008, Lupo008, Chruscinski2014, Shen2015, Fan2002}. The infinite dimensional applications are also possible and are reported in a separate note \cite{RGielerakFuture4}.

With  the  help  of  these  expansions, the  following formulas   for  the  corresponding  reduced  density  operators (RDM)
on the  local $L_2$-spaces  are  derived

\begin{corollary}
Let  $\mathcal{H}= \mathcal{H}_A \otimes \mathcal{H}_B$  be  a  bipartite  separable  Hilbert  space and let
$Q \in E (\mathcal{H})$. Then $L_2$-RDM  of $|Q \rangle \langle Q| \in L_2 (L_2 (\mathcal{H}))$ are  given  by
\begin{equation}
QQ^A = \mathrm{Tr}_{L_2(\mathcal H_B)}  ( |Q \rangle \langle Q| ) = 
\sum_{n=1}^{\infty} \tau_n^2 |\Omega_n^A \rangle \langle\Omega_n^A| ,
\end{equation}
and
\begin{equation}
QQ^B = \mathrm{Tr}_{L_2(\mathcal H_A)}  ( |Q \rangle \langle Q| ) = 
\sum_{n=1}^{\infty} \tau_n^2 |\Omega_n^B \rangle \langle\Omega_n^B| .
\end{equation}
\end{corollary}

The  hermitean version  of  this  expansion  is    given:
\begin{corollary}
Let  $\mathcal{H}= \mathcal{H}_A \otimes \mathcal{H}_B$  be  a  bipartite  separable  Hilbert  space and let
$Q \in E (\mathcal{H})$. Then $He L_2$-RDM  of $|Q \rangle \langle Q| \in He (He (\mathcal{H}))$ are  given  by
\begin{equation}
QQ^A = \mathrm{Tr}_{L_2(\mathcal H_B)} ( |Q \rangle \langle Q| ) = 
\sum_{n=1}^{\infty} \tau_n^{\star 2} |\Omega_n^{\star  A} \rangle \langle \Omega_n^{\star  A}| ,
\end{equation}
and
\begin{equation}
QQ^B = \mathrm{Tr}_{L_2(\mathcal H_A)} (|Q \rangle \langle Q| ) = 
\sum_{n=1}^{\infty} \tau_n^{\star 2} |\Omega_n^{\star B} \rangle \langle \Omega_n^{\star  B}| .
\end{equation}
\end{corollary}

The  operator  $| Q \rangle \langle Q |$   acts  in  the  Hilbert-Schmidt  space  of  operators acting in $\mathcal{H}$  as  an orthogonal
projector. The  spaces  of  operators  acting  on  the  space  of  states $E(\mathcal{H})$ are  called  often  the space of
superoperators. From  the  physical point of  view  the  most  important  class of  superoperators  are  those  which  are
completely  positive and  trace  preserving  \cite{NC, BZ, AlickiLendi2007}. Such  superoperators  are  called  quantum  channels. From  our
considerations  it  follows  that any  superoperator from $He2 ( He2(\mathcal{H}))$ can  be  decomposed similarly  to  the
decompositions (B.5)-(B.8). 

%These  decompositions applied  to quantum  channels in  particular,  and  its  connections  to the  canonical Kraus  decompositions are being  analysed in several  papers, see \cite{X1,X2,X3}\Hm{citations}. Also the question whether   the  derived  in Corollary B.3  and Corollary  B.4\Hm{good ref. num.?} expansions find useful  applications looks to be not explored yet.

%\section{Appendix C. Operator  valued (renormalized)  map ($Q \rightarrow  \log(\mathrm{I}_{\mathcal H} + Q))$}
\section{Operator  valued (renormalized)  map ($Q \rightarrow  \log(\mathrm{I}_{\mathcal H} + Q))$} \label{lbl:app:operator:renorm:map}

\setcounter{equation}{0}

Several useful properties  of  the  map $Q \rightarrow  \log(\mathrm{I}_{\mathcal H} +Q)$  will be  collected in  this  supplement.
To  start  with  let  us  consider  non-negative  $Q  \in L(\mathcal{H})$. Using  the  spectral  theorem we  can  define
operator  $\log(\mathrm{I}_{\mathcal H} +Q)$.

\begin{proposition}
The  map  $\log(1_\mathcal{H} +.)$ with values in $L_1^{+}(\mathcal{H})$ is well defined on $L_1^{+}(\mathcal{H})$
and moreover, for $Q  \in L_1^{+}(\mathcal{H})$:
\begin{enumerate}
  \item $\| \log (1_\mathcal{H} + Q) \|_1 \leq \|Q \|_1$,
  \item the  map  $Q \rightarrow  \log(1_\mathcal{H} +Q)$  is  operator  monotone map,
  \item The  map  $\log(1_\mathcal{H} \, + \, .)$ as  defined  on the  cone $L_1^{+}(\mathcal{H})$  is  strictly operator concave function which means the following for any $Q_0, Q_1 \in L_1^{+}(\mathcal{H})$, any $\tau \in (0,1)$:
\begin{equation}
\log (1_\mathcal{H} + \tau Q_1 + (1-\tau) Q_2) \geq \tau \log (1_\mathcal{H} + Q_1)  + (1 - \tau) \log (1_\mathcal{H} + Q_2) .
\end{equation}
\end{enumerate}
For any $Q \in E(\mathcal{H})$:
\begin{equation}
\mathrm{Tr}[\log ( \mathrm{1}_\mathcal H+Q) ] \leq  1 .
\end{equation}
\end{proposition}

\begin{lemma}[C.2]
For  any $Q_0, Q_1 \in L_1^{+}(\mathcal{H})$  the  strong  Frechet  directional derivative  of  the  map
$\log (1_H+ . .)$  in the  point  $Q_0$ and in the  direction  to $Q_1$ is  given  by  the  formula:
\begin{equation}
\nabla Q_1(\log (1_H + \ldots)(Q_0 ) = \int_0^{\infty} dx (1_h + Q_0 + x)^{-1}Q_1(1_H + Q_) +x)^{-1} .
\end{equation}
\end{lemma}

\begin{theorem}[C]
For  any $Q_0, Q_1 \in L_1^{+}(\mathcal{H})$  the  following  estimate  is  valid
\begin{equation}
  \| \log (1_H + Q_1 ) – \log ( 1_H + Q_2 ) \|_1 \leq     o(1) \| Q_1 - Q_2 \|_1 .
\end{equation}
\end{theorem}
\begin{proof}
All the formulated here results are valid in the finite dimensional setting.  The corresponding infinite dimensional results follows by performing the finite dimensional approximations and then performing the passage (controllable by the $L_1$-continuity) to limiting cases.
\end{proof}

\subsection{Continuation of the proof of Theorem~\ref{th-cont}}

The  case  $\tau=1$.

If  $\tau=1$  then it  follows  that  $Q$  or  $Q'$ or  both one are pure states. Assume  that  $Q$,  $Q'$  are  both  pure  states. %Assume that $Q, Q' \in \partial E(\mathcal{H})$. 
Then, there  exist  two unit  vectors  $| \psi \rangle$  and  $| \theta \rangle$ such that $Q=| \psi \rangle \langle \psi |$  and  $Q' = | \theta \rangle \rangle \theta |$. From  the  idempotency of  $Q$  and  $Q'$  it follows :
\begin{equation}
	\log(1+Q)= \log(2) \cdot | \psi \rangle \langle \psi | ,
\end{equation}
and 
\begin{equation}
	\log(1+Q')= \log(2) \cdot | \theta \rangle \rangle \theta | ,
\end{equation}
from which
\begin{equation}
{\| \log( 1 + Q ) - \log( 1 + Q') \|}_{1} \leq o(1) \| \psi  -  \theta  \| .
\end{equation}	
If $Q'$ is not pure but $Q \in \partial E ( \mathcal{H} ) $ then $ \delta = \| Q' \| < 1$. Taking $Q'$  such  that $\| Q '- |\psi\rangle \langle \psi | \| = \rho \leq  \delta$  the  proof  follows  by repeating almost literally the arguments as in (\ref{eq-ineq1}).

\section*{Acknowledgments}

It is my pleasure  to  thank  dr  Sylvia  Kondej  for  her  patience with  reading a very preliminary version of  the  present  paper  and  in  particular  for making  several corrections of   different nature  that have caused  significant  improvements  to  the present version  of it. Additionally  her  work  together  with  dr  Marek Sawerwain  to  successfully  convert  the original  .doc  version of this manuscript  into  the present   Tex  version  is  very much  appreciated  by the  Author. The  kind  hospitality  of the   Institute of Control \& Computation Engineering,University of Zielona G\'ora, in particular   that of  prof. Józef  Korbicz  was  very  helpful in the  time  of  writing  these notes .

\end{document}